\documentclass[%
 reprint,
superscriptaddress,
 amsmath,amssymb,
 aps,
pra,
]{revtex4-2}
\usepackage[utf8]{inputenc}

\usepackage[caption=false]{subfig}
\usepackage{graphicx}
\usepackage{here}
\usepackage{epstopdf}
\usepackage{array}
\usepackage{physics}
\usepackage{verbatim}
\usepackage{amsmath,amsfonts,amssymb,amscd,mathtools}
\usepackage{amsthm}
\usepackage{tabularx}
\usepackage{stmaryrd}
\usepackage{enumerate}
\usepackage{wasysym}
\usepackage{braket}
\usepackage{mathrsfs}
\usepackage{microtype}
\usepackage{hyperref}
\usepackage{soul}
\usepackage[dvipsnames]{xcolor}
\hypersetup{
    bookmarksnumbered=true, % If Acrobat bookmarks are requested, include section numbers
    unicode=false, % non-Latin characters in Acrobat bookmarks
    pdfstartview={FitH}, % fits the width of the page to the window
    pdftitle={}, % title
    pdfauthor={}, % author
    pdfsubject={}, % subject of the document
    pdfcreator={}, % creator of the document
    pdfproducer={}, % producer of the document
    pdfkeywords={}, % list of keywords
    pdfnewwindow=true, % links in new window
    colorlinks=true, % false: boxed links; true: colored links
    linkcolor=NavyBlue, % color of internal links
    citecolor=NavyBlue, % color of links to bibliography
    filecolor=NavyBlue, % color of file links
    urlcolor=NavyBlue % color of external links
}
\usepackage{tikz}
\usepackage[lmargin=.7in,rmargin=.7in,tmargin=.7in,bmargin=1in]{geometry}
\usepackage{newtxtext,newtxmath}

\usepackage{cleveref}

\usepackage{algorithm}
\usepackage{algorithmicx}
\usepackage{algpseudocode}


\usepackage{booktabs}

\theoremstyle{plain}
\newtheorem{thm}{Theorem}

\newtheorem{lem}[thm]{Lemma}
\newtheorem{pro}[thm]{Proposition}

\theoremstyle{definition}
\newtheorem{defn}[thm]{Definition}

\usepackage[most,breakable]{tcolorbox}
\tcbset{ sharp corners}
  {\expandafter\ifstrequal\expandafter{#1}{filled}{\begin{tcolorbox}[colback=MidnightBlue!70!black!70!TealBlue!2!white,colframe=MidnightBlue!70!black!70!TealBlue!30!white,breakable,enhanced,left=5.75pt,right=5.75pt,grow sidewards by=10pt]}{\begin{tcolorbox}[colback=white,colframe=gray!0,breakable,enhanced,left=5.75pt,right=5.75pt,grow sidewards by=10pt]}}%
  {\end{tcolorbox}}

\newenvironment{defnboxed}[1][white]%definition
  {\expandafter\ifstrequal\expandafter{#1}{filled}{\begin{tcolorbox}[colback=MidnightBlue!70!black!70!TealBlue!2!white,colframe=MidnightBlue!70!black!70!TealBlue!30!white,breakable,enhanced,left=5.75pt,right=5.75pt,grow sidewards by=10pt]}{\begin{tcolorbox}[colback=MidnightBlue!10,colframe=MidnightBlue!0,breakable,enhanced,left=5.75pt,right=5.75pt,grow sidewards by=10pt]}}%
  {\end{tcolorbox}}

\newenvironment{proboxed}[1][white]%proposition
  {\expandafter\ifstrequal\expandafter{#1}{filled}{\begin{tcolorbox}[colback=MidnightBlue!70!black!70!TealBlue!2!white,colframe=MidnightBlue!70!black!70!TealBlue!30!white,breakable,enhanced,left=5.75pt,right=5.75pt,grow sidewards by=10pt]}{\begin{tcolorbox}[colback=OliveGreen!10,colframe=OliveGreen!0,breakable,enhanced,left=5.75pt,right=5.75pt,grow sidewards by=10pt]}}%
  {\end{tcolorbox}}

  \newenvironment{thmboxed}[1][white]%Theorem
  {\expandafter\ifstrequal\expandafter{#1}{filled}{\begin{tcolorbox}[colback=MidnightBlue!70!black!70!TealBlue!2!white,colframe=MidnightBlue!70!black!70!TealBlue!30!white,breakable,enhanced,left=5.75pt,right=5.75pt,grow sidewards by=10pt]}{\begin{tcolorbox}[colback=Maroon!10,colframe=Maroon!0,breakable,enhanced,left=5.75pt,right=5.75pt,grow sidewards by=10pt]}}%
  {\end{tcolorbox}}
  
   \newenvironment{lemboxed}[1][white]%Theorem
  {\expandafter\ifstrequal\expandafter{#1}{filled}{\begin{tcolorbox}[colback=MidnightBlue!70!black!70!TealBlue!2!white,colframe=MidnightBlue!70!black!70!TealBlue!30!white,breakable,enhanced,left=5.75pt,right=5.75pt,grow sidewards by=10pt]}{\begin{tcolorbox}[colback=Gray!20,colframe=Maroon!0,breakable,enhanced,left=5.75pt,right=5.75pt,grow sidewards by=10pt]}}%
  {\end{tcolorbox}}

\newcommand{\eq}[1]{(\hyperref[eq:#1]{\ref*{eq:#1}})}

\renewcommand{\sec}[1]{\hyperref[sec:#1]{Section~\ref*{sec:#1}}}
\newcommand{\thrm}[1]{\hyperref[thrm:#1]{Theorem~\ref*{thrm:#1}}}
\newcommand{\lemm}[1]{\hyperref[lemm:#1]{Lemma~\ref*{lemm:#1}}}
\newcommand{\prop}[1]{\hyperref[prop:#1]{Proposition~\ref*{prop:#1}}}
\newcommand{\corr}[1]{\hyperref[corr:#1]{Corollary~\ref*{corr:#1}}}
\newcommand{\fig}[1]{\hyperref[fig:#1]{~\ref*{fig:#1}}}
\newcommand{\deff}[1]{\hyperref[deff:#1]{~\ref*{deff:#1}}}

\newcommand{\mA}{\mathcal{A}}

\newcommand{\mE}{\mathcal{E}}

\newcommand{\mD}{\mathcal{D}}

\newcommand{\mL}{\mathcal{L}}

\newcommand{\mH}{\mathcal{H}}

\newcommand{\mB}{\mathcal{B}}

\newcommand{\mP}{\mathcal{P}}

\newcommand{\mR}{\mathcal{R}}

\newcommand{\mbD}{\mathbb{D}}

\newcommand{\mbF}{\mathbb{F}}

\newcommand{\mbN}{\mathbb{N}}
\newcommand{\mbO}{\mathbb{O}}
\newcommand{\mbR}{\mathbb{R}}

\newcommand{\mfR}{\mathfrak{R}}

\newcommand{\saf}{\widetilde{f}}
\newcommand{\pf}{\overline{f}}
\newcommand{\sQ}{\widetilde{Q}}
\newcommand{\pQ}{\overline{Q}}

\newcommand{\xm}{{X_m}}

\newcommand{\sD}{\widetilde{D}}
\newcommand{\sH}{\widetilde{H}}
\newcommand{\pH}{\overline{H}}
\newcommand{\pD}{\overline{D}}

\newcommand{\spsi}{\widetilde{\psi}}

\renewcommand{\*}{\textup{*}}
\DeclareMathOperator{\cone}{cone}

\newcommand{\ve}{\varepsilon}

\DeclareMathOperator{\aff}{aff}

\DeclareMathOperator{\Span}{Span}
\DeclareMathOperator{\supp}{supp}

\DeclareMathOperator{\spec}{spec}

\DeclareMathOperator{\argmax}{argmax}

\newcommand{\ba}{\begin{eqnarray}}
\newcommand{\ea}{\end{eqnarray}}
\newcommand{\bann}{\begin{eqnarray*}}
\newcommand{\eann}{\end{eqnarray*}}
\newcommand{\bal}{\begin{equation}\begin{aligned}}
\newcommand{\eal}{\end{aligned}\end{equation}}
\newcommand{\dm}[1]{\ketbra{#1}{#1}}

\newcolumntype{L}[1]{>{\raggedright}p{#1}}
\newcolumntype{C}[1]{>{\centering}p{#1}}
\newcolumntype{R}[1]{>{\raggedleft}p{#1}}
\newcolumntype{D}{>{\centering\arraybackslash}X}

\renewcommand{\*}{\textup{*}}

\usepackage{dcolumn}% Align table columns on the decimal point
\usepackage{bm}% bold math
\makeatletter
\renewcommand{\paragraph}[1]{\medskip\noindent\textbf{\textit{#1}}\ignorespaces}
\makeatother

\begin{document}

\title{Reliability of asymptotic work extraction}% Force line breaks with \\

\author{Kaito Watanabe}
\email{watanabe715@g.ecc.u-tokyo.ac.jp}
\affiliation{Department of Basic Science, The University of Tokyo, 3-8-1 Komaba, Meguro-ku, Tokyo 153-8902, Japan}
\affiliation{Mathematical Quantum Information RIKEN Hakubi Research Team, RIKEN Pioneering Research Institute (PRI) and RIKEN Center for Quantum Computing (RQC), Wako, Saitama 351-0198, Japan}

\author{Bartosz Regula}
\email{bartosz.regula@gmail.com}
\affiliation{Mathematical Quantum Information RIKEN Hakubi Research Team, RIKEN Pioneering Research Institute (PRI) and RIKEN Center for Quantum Computing (RQC), Wako, Saitama 351-0198, Japan}

\author{Marco Tomamichel}
\email{marco.tomamichel@nus.edu.sg}
\affiliation{Centre for Quantum Technologies, National University of Singapore, Singapore 117543, Singapore}
\affiliation{Department of Electrical and Computer Engineering,
National University of Singapore, Singapore 117583, Singapore}  

\author{Ryuji Takagi}
\email{ryujitakagi@g.ecc.u-tokyo.ac.jp}
\affiliation{Department of Basic Science, The University of Tokyo, 3-8-1 Komaba, Meguro-ku, Tokyo 153-8902, Japan}

%TC:ignore
\begin{abstract}
Extracting work from quantum states is a fundamental task in quantum thermodynamics. Previous studies have primarily focused on determining the best achievable rate of work extraction, and remarkably, this characterization appeared to remain unchanged regardless of the choice of allowed processes: whether one considers the operationally motivated class of energy-conserving thermal operations, or the axiomatic class of Gibbs-preserving operations, the optimal extractable work is given by the Helmholtz free energy.
Here, we challenge this perspective, showing that a more refined analysis of the asymptotic performance of work extraction reveals significant differences in the performance for the two different classes of free operations.
Precisely, we focus on the trade-off between the extraction rate and its reliability, characterized by the optimal asymptotic speed at which the extraction error can be suppressed. 
We establish that the reliability of Gibbs-preserving operations and of thermal operations are respectively characterized by the Petz and the sandwiched R\'enyi relative entropies, demonstrating that the former in general strictly outperforms the latter, and providing new interpretations of several information-theoretic divergences.
Our analysis reveals that operational constraints such as energy conservation impose stronger limitations on the achievable precision of quantum tasks than can be inferred from their asymptotic rates, thereby questioning the use of Gibbs-preserving operations as a mathematically convenient substitute for the physically realizable thermal processes. 

\end{abstract}
%TC:endignore
\maketitle

\paragraph{Introduction.}\,---\, Characterizing the ultimate performance of manipulating nonequilibrium quantum states is a central problem in quantum thermodynamics~\cite{horodecki_2013, Faist_2018_fundamental_work, Lostaglio_2019_introductory, Brandao_resource_theory_of_quantum, Gour_role_of_quantum_coherence}, as it determines the fundamental limits of work extraction, one of the most basic notions in thermodynamics. 
With recent advances in the precise control of nanoscale quantum systems, understanding such ultimate limits has become increasingly important.

A powerful framework for addressing this problem is the resource-theoretic approach~\cite{horodecki_2013,Brandao_resource_theory_of_quantum}, where thermodynamic state transformations are studied under a restricted class of quantum operations, called free operations~\cite{Chitamber_Gour, Gour_2025_book}. 
A particularly useful insight from this approach is that work extraction can be quantitatively connected to quantum hypothesis testing, establishing a direct bridge between thermodynamics and information theory.

Within the resource-theoretic formalism, one makes a suitable choice of the allowed free operations, aiming to reflect the operational constraints of the given setting.
Among the possible such choices, two representative classes are the axiomatically-motivated \emph{Gibbs-preserving operations} --- which are a useful approximation of the operational restrictions of quantum thermodynamics, often chosen due to their simple structure --- and the physically implementable \emph{thermal operations}, which can be easily realized without expending any thermodynamic resources, but whose more complicated mathematical structure leads to difficulties in their characterization.
While these two classes are known to exhibit significant structural differences~\cite{Faist2015Gibbs-preserving, Tajima_2024_Gibbs-preserving}, previous studies showed that the optimal asymptotic work extraction rate under both classes is universally characterized by the Helmholtz free energy. 
This naturally suggests an operational equivalence between the two classes in the task of work extraction, leading to the impression that the mathematically convenient choice of Gibbs-preserving operations leads to no loss of physical insight.

In this work, we show that this intuition breaks down once one considers not 
merely the amount of extracted work, but rather the error exponent (also called the reliability function) of work extraction. 
Specifically, we ask: if one aims to achieve a target work extraction rate, how fast can the extraction error vanish asymptotically? 
This refinement gives a much more precise understanding of the performance that can be achieved in practice, and it is a natural question to consider in many relevant scenarios\,---\,errors are often required to vanish sufficiently fast to allow for reliable work extraction.

Our main result reveals a sharp operational separation between Gibbs-preserving operations and thermal operations: although both classes exhibit the same performance when reliability is not taken into account, their achievable rates differ when the error is required to vanish exponentially fast.
In other words, the error exponents achievable under Gibbs-preserving operations are in general strictly larger than those achievable under thermal operations for the same extraction rate.
The result relies on the analysis of the underlying error exponents of quantum hypothesis testing, allowing us to characterize the asymptotic error decay at any fixed extraction rate.
Remarkably, this separation is completely characterized by distinct information-theoretic quantities. 
Under Gibbs-preserving operations, the error exponent is governed by the \emph{Petz Rényi relative entropy}~\cite{Petz_1986_quasi_entropy}, whereas under thermal operations it is governed by the \emph{sandwiched Rényi relative entropy}~\cite{Muller_Lennert_2013_on_quantum_renyi, Wilde_2014_strong_converse}. 
Our finding thus uncovers a fundamental gap in asymptotic performance between the mathematically convenient axiomatic free operations and the physically implementable thermodynamic processes, putting into question the suitability of using Gibbs-preserving operations as a faithful thermodynamic model.

The performance gap is particularly prominent in the setting of zero-rate work extraction, which places the focus purely on the quality of work extraction rather than its yield. 
Characterizing this setting completely, we show that while the reliability under Gibbs-preserving operations is governed by a reverse variant of the conventional quantum relative entropy, under thermal operations it is given by the \emph{star divergence}~\cite{Audenaert_alpha_z_2015, Lipka_Bartosik_2024} --- a limit of the so-called reverse sandwiched R\'enyi divergences~\cite{Audenaert_alpha_z_2015} --- which has not found applications in thermodynamic work extraction before.

Altogether, our results provide new operational interpretations of Petz Rényi, sandwiched Rényi, and the star divergences in thermodynamics. 
Notably, we find that even divergence-like quantities that do not satisfy the standard data-processing inequality --- such as sandwiched Rényi divergence with $\alpha<1/2$, or the star divergence --- find use in characterizing the reliability of practical tasks.
We show in particular that both quantities nevertheless satisfy monotonicity under time-translation-covariant channels, thereby identifying symmetry-restricted monotonicity as the operational principle underlying the applications of these quantities.

Quantum thermodynamics can be viewed as a special case of the broader framework of quantum resource theories~\cite{Chitamber_Gour, Gour_2025_book}, which study the ultimate limits of manipulating valuable quantum resources. 
More generally, we show that the correspondence between the precision of resource manipulation and hypothesis testing extends beyond thermodynamics to general quantum resource theories, including tasks such as resource distillation and dilution, and yields bounds and exact characterizations of the corresponding error and strong-converse exponents. 
The general framework, together with the proofs of the main results presented here, is deferred to the Appendix.

\paragraph{Quantum thermodynamics.}\,---\,
We consider a quantum system associated with a Hilbert space $\mH$ and Hamiltonian $H$, in contact with a thermal bath of inverse temperature $\beta$.

To analyze both the amount and the precision of work extraction, we employ the resource-theoretic approach, which has been successful in characterizing the fundamental limitations of various quantum resource manipulation tasks under operational constraints~\cite{Chitamber_Gour, Gour_2025_book}.
In general, a quantum resource theory is specified by identifying the class of quantum channels that can be applied at no cost (called \emph{free operations}) and the sets of states that can be prepared freely (called \emph{free states}). 
In quantum thermodynamics, the set of free states is the singleton $\qty{\tau}$ containing the Gibbs thermal state $\tau=e^{-\beta H}/\Tr[e^{-\beta H}]$, since the thermal state can always be prepared by leaving the system in contact with the thermal bath until it thermalizes.

One standard choice of free operations is the \emph{thermal operations} (TO), which correspond to physically implementable isothermal processes. 
Let $\mD(\mH)$ denote the set of density operators acting on the Hilbert space $\mH$.
A quantum channel $\Lambda:\mD(\mH_A)\to\mD(\mH_B)$ is called a thermal operation if and only if $\Lambda$ is decomposed as (1) appending the thermal state of an ancillary system, (2) applying an energy-conserving unitary, and (3) discarding a subsystem. Mathematically, a thermal operation is a channel which is decomposed as $\Lambda(\rho_A)=\Tr_{A+E-B}[U_{AE}(\rho_A\otimes \tau_E)U_{AE}^{\dagger}]$, where $U_{AE}$ satisfies $[U_{AE},H_A\otimes I_E +I_A\otimes H_E ]=0.$ 
We call a channel realized as a limit of a sequence of thermal operations also a thermal operation.
Thermal operations satisfy two important properties. One is that any thermal operation preserves the thermal state, that is, $\Lambda(\tau_A)=\tau_B$ holds. The other is that thermal operations are time-translation covariant, i.e., it holds that $\Lambda(e^{-iH_At}\rho_A e^{iH_A t})=e^{-iH_Bt}\Lambda(\rho_A)e^{iH_Bt}.$

Since the class of thermal operations is mathematically difficult to describe, one often considers a larger class of operations called the \emph{Gibbs-preserving operations} (GPO)~\cite{Janzing_2000_thermodynamic, Faist_2018_fundamental_work}. Here, the class of Gibbs-preserving operations consists of all quantum channels satisfying $\Lambda(\tau_A)=\tau_B$. 
Although Gibbs-preserving operations have a simple mathematical structure and are easier to analyze, they are known to be strictly more powerful than thermal operations. In fact, they can create and detect energetic coherence from scratch~\cite{Faist2015Gibbs-preserving}, which is prohibited by time-translation covariance, and their implementation by thermal operations may require an infinite amount of coherence as an additional resource~\cite{Tajima_2024_Gibbs-preserving}.

We are now in a position to discuss how to quantify the amount and precision of work extraction.
To quantify the work extracted from a quantum state, we consider an additional quantum system, which we call a \emph{work battery}~\cite{Wang_2019, Gour_role_of_quantum_coherence}. Work battery $X_m$ labeled by a positive number $m>0$ is associated with a 2-dimensional Hilbert space $\mH_{X_m}=\Span\qty{\ket{0}_{X_m}, \ket{1}_{X_m}}$. We take a Hamiltonian $H_{X}=E_{X,0}\dm{0}+E_{X,1}\dm{1}$ with $E_{X,1}-E_{X,0}=\beta^{-1}\log(m-1)$ so that the thermal state of the battery system $\mu_m$ is $\mu_m=\frac{m-1}{m}\ketbra{0}{0}+\frac{1}{m}\ketbra{1}{1}$.
Suppose that we are given a quantum state $\rho\in\mD(\mH)$. If the conversion $\rho\otimes \mu_m\mapsto\ketbra{1}{1}_\xm$ is possible up to some fidelity error $\ve$ under some free operation $\Lambda\in\mbO$, we say that the work $\beta^{-1}\log m$ is extractable from the quantum state $\rho$ up to fidelity error $\ve$.
The one-shot optimal extractable work from $\rho\in\mD(\mH)$ under the class of free operations $\mbO$ is defined as $\beta W^\ve_{\mbO}(\rho):=\log\max \qty{m>0~|~\max_{\Lambda\in\mbO}\bra{1}\Lambda(\rho\otimes \mu_m)\ket{1}\geq 1-\ve}$.

We can also consider the asymptotic (thermodynamic) limit of this quantity by taking multiple copies of the same system. Here, following the standard assumption~\cite{Brandao_resource_theory_of_quantum}, we suppose that there is no correlation between subsystems, and the Hamiltonian of the whole system is described as $H^{\times n}:=\sum_{i=1}^nI^{\otimes (i-1)}\otimes H\otimes I^{\otimes (n-i)}$.
The asymptotic extractable work per copy under the class $\mbO$ of free operations is defined as $\beta W^{\rm asymp}_\mbO(\rho):=\lim_{\ve\to +0}\limsup_{n\to \infty}\frac{1}{n}\beta W^\ve_\mbO(\rho^{\otimes n})$.

Previous studies on work extraction showed that the asymptotic extractable work from i.i.d.\ states is characterized as $\beta W^{\rm asymp}_\mbO(\rho)=D(\rho\|\tau)$ under both thermal operations and Gibbs-preserving operations.
Here, $D(\rho\|\tau):=\Tr[\rho\log \rho-\rho\log \tau]$ is the Umegaki relative entropy~\cite{Umegaki_relative}, which is directly connected to the Helmholtz free energy since $D(\rho\|\tau)$ satisfies $D(\rho\|\tau)=\Tr[\rho H]-\beta^{-1}S(\rho)-F_{\rm eq}$, where $S(\rho):=-\Tr(\rho\log\rho)$ is the von Neumann entropy and $F_{\rm eq}:=-\log\Tr\qty[e^{-\beta H}]/\beta$ is the equilibrium free energy. 
Thus, at the first-order asymptotic level, whether or not we impose time-translation covariance on the available operations does not affect the amount of extractable work. 
The Umegaki relative entropy arises from the tight connection between work extraction and the state-discrimination task known as quantum hypothesis testing.

\paragraph{Quantum hypothesis testing.}\,---\,%
We briefly review standard notions in quantum hypothesis testing, an information-theoretic task of discriminating between two hypotheses: the \emph{null hypothesis} $\rho\in\mD(\mH)$ and the \emph{alternative hypothesis} $\sigma\in\mD(\mH)$. 
Suppose that either $\rho$ or $\sigma$ is given, and our goal is to determine which state is actually given by performing a test described by a POVM $\qty{M, I-M}$. 
Here, $M$ corresponds to the measurement outcome associated with the guess that the given state is $\rho$, while $I-M$ corresponds to the opposite guess.

There are two possible types of error in this task. 
One is to infer that $\sigma$ is given when the true state is $\rho$, called the type I error. 
The other is to infer that $\rho$ is given when the true state is $\sigma$, called the type II error. 
For a test $\qty{M, I-M}$, the probability of type I error is $\Tr[\rho(I-M)]$, and that of type II error is $\Tr[\sigma M]$.

A standard setting is the asymmetric hypothesis testing scenario, where one minimizes the probability of type II error under the constraint that the probability of type I error is no larger than a fixed constant $\ve>0$. 
The corresponding figure of merit, called the hypothesis testing divergence, is defined as %$D^\ve_H(\rho\|\sigma)=-\log \inf_{\substack{0\leq M\leq I\\ \Tr[\rho(I-M)]\leq \ve }}\Tr[\sigma M]$.
$D^\ve_H(\rho\|\sigma)=-\log \inf \{ \Tr[\sigma M] : 0\leq M\leq I, \Tr[\rho(I-M)]\leq \ve \}$.

When multiple copies of $\rho$ or $\sigma$ are available, the discrimination task becomes easier, and the error probabilities decay exponentially with the number of copies. 
In fact, previous works has identified the optimal asymptotic decay rates of these errors, known as Stein's exponent, as the Umegaki relative entropy as~\cite{hiai_1991_proper, Ogawa_2000_strong} $\lim_{n\to\infty}\frac{1}{n}D^\ve_H(\rho^{\otimes n}\|\sigma^{\otimes n})=D(\rho\|\sigma),~~\forall\ve\in(0,1)$.

Another important setting is to characterize the smallest $\ve$ such that the probability of type II error is kept below a prescribed target value. 
In the asymptotic regime, the optimal type I error vanishes when the target decay rate of the type II error is below the optimal rate, whereas it converges to one when the target decay rate exceeds the optimal rate. 
These are known as the error exponent regime and the strong converse exponent regime, respectively.

The exponent governing the decay of the type I error in the former regime is called the error exponent, and is fully characterized by a quantum extension of the Rényi divergence, called the Petz Rényi divergence~\cite{Petz_1986_quasi_entropy} $\pD_\alpha(\rho\|\sigma):=\frac{1}{\alpha-1}\log \Tr[\rho^\alpha\sigma^{1-\alpha}]$~\cite{nagaoka_2006_converse_theorem, Hayashi_2007_error_exponent}. Another common quantum extension of Rényi divergence that plays an important role in our work is called the sandwiched Rényi divergence~\cite{Muller_Lennert_2013_on_quantum_renyi, Wilde_2014_strong_converse} $\sD_\alpha(\rho\|\sigma):=\frac{1}{\alpha-1}\log \Tr[(\sigma^{\frac{1-\alpha}{2\alpha}}\rho\sigma^{\frac{1-\alpha}{2\alpha}})^\alpha]$. Although it also finds use in characterizing quantum hypothesis testing, it typically emerges in the complementary regime of strong converse exponents~\cite{Mosonyi_Ogawa_2014,Hayashi_2016} and not in the analysis of error decay.
The fact that there is seemingly not a single type of Rényi divergence that governs the asymptotic behavior of quantum systems has, in addition to causing considerable technical difficulties~\cite{ogawa_2004,Hayashi_2007_error_exponent}, led to persistent conceptual unease~\cite{tomamichel_2025}.

\paragraph{One-shot performance of work extraction.}\,---\,%
We now characterize the optimal precision of work extraction, quantified by the fidelity error with respect to the target work state.
Suppose that we are given a quantum state $\rho$. 
For a target work $W$ (equivalently, the dimensionless quantity $\beta W$), our goal is to minimize the fidelity error of the final state. 
The one-shot optimal error of work extraction from $\rho$ with target work $W$ under a class $\mbO$ of free operations is defined as the infidelity between the output state and the target excited battery state
\bal
\mE_{\mbO}(\rho; \beta W):=1-\max_{\Lambda\in\mbO}F(\Lambda(\rho\otimes \mu_{\beta W}),\ketbra{1}{1}_{X_{\beta W}}).
\eal

We begin by characterizing this one-shot optimal error.
The following result shows that this optimization can be reduced to a hypothesis-testing problem between the input state $\rho$ and the thermal state $\tau$.

 \begin{pro}~\label{pro: one-shot error of work extraction}
    The one-shot optimal error of work extraction from $\rho$ with target amount $W$ of work under Gibbs-preserving operations and thermal operations is given by
    \bal
    \mE_{{\rm GPO}}(\rho; \beta W)&=\min\qty{\ve~|~\beta W\leq D^\ve_H(\rho\|\tau)},\\
    \mE_{{\rm TO}}(\rho; \beta W)&=\min\qty{\ve~|~\beta W\leq D^\ve_H(\mP_{\tau}(\rho)\|\tau)}.
    \eal
    Here, $\mP_\tau$ is the pinching channel~\cite{Hayashi_2002_optimal_sequence} defined as $\mP_\tau(\cdot):=\sum_{i}\Pi_i(\cdot)\Pi_i$, where $\tau=\sum_i\lambda_i\Pi_i$ is the spectral decomposition of $\tau$.
\end{pro}

Proposition~\ref{pro: one-shot error of work extraction} shows that the optimal error of work extraction is fully characterized by the smallest error threshold $\ve$ satisfying the hypothesis-testing constraint. 
An immediate consequence is that the more distinguishable $\rho$ and $\tau$ are, the more reliably one can extract work from $\rho$.

We note that connections between work extraction and one-shot entropies such as $D^\ve_H$ have a long history in quantum information~\cite{dahlsten_2011,horodecki_2013,yungerhalpern_2016}. A key aspect that we rely on is not merely an approximate, but an \emph{exact} correspondence between work extraction and hypothesis testing --- for Gibbs-preserving operations, this is a consequence of the prior findings of~\cite{Wang_2019,Buscemi2019information,Regula_benchmarking}; the expression in the second equality of Proposition~\ref{pro: one-shot error of work extraction} also appeared in~\cite{Gour_role_of_quantum_coherence} in the study of Gibbs-preserving covariant operations, an intermediate class between GPO and TO.

Since thermal operations are more restrictive than Gibbs-preserving operations, the achievable reliability of work extraction is expected to decrease. 
Proposition~\ref{pro: one-shot error of work extraction} shows that this loss of reliability is precisely captured by the appearance of the pinching channel $\mP_\tau$. 
Indeed, by the data-processing inequality of the hypothesis testing divergence~\cite{Wang_2012_oneshot}, it follows that $\mE_{{\rm GPO}}(\rho; \beta W)\leq \mE_{{\rm GPO}}(\mP_\tau(\rho); \beta W)=\mE_{{\rm TO}}(\rho; \beta W)$.
That is, thermal operations effectively replace the input state by its dephased version in the energy eigenbasis, and the resulting loss of distinguishability directly limits the achievable precision.

\paragraph{Error exponent of work extraction.---}
We characterize the asymptotic behavior of the optimal error of work extraction at a fixed target work extraction rate. 
Our main focus is on the error exponent of work extraction, which quantifies the exponential rate at which the optimal error of work extraction vanishes.
For a fixed target work extraction rate $r$, it is defined as 
\begin{equation}
    B_{\mbO}(\rho; r):=\sup\left\{
    \begin{aligned}
        &\liminf_{n\to\infty}-\frac{1}{n}\log \mE_{\mbO}(\rho^{\otimes n},\beta W_n)\\
&~|~\liminf_{n\to\infty}\frac{1}{n}\beta W_n\geq r
    \end{aligned}\right\}.
\end{equation}

Proposition~\ref{pro: one-shot error of work extraction} allows us to directly relate these quantities to the corresponding error exponent in quantum hypothesis testing. 
In particular, the error exponent of work extraction is associated with the largest $c\geq 0$ such that the type I error decays exponentially as $2^{-cn}$, while the type II error decays no slower than $2^{-n\beta W}$.

In particular, when thermal operations are taken as the class of free operations, the problem reduces to hypothesis testing between the sequence of pinched states $\qty{\mP_{\tau^{\otimes n}}(\rho^{\otimes n})}_{n\in\mbN}$ and the thermal states $\qty{\tau^{\otimes n}}_{n\in\mbN}$. 
Exploiting the seminal derivations of the error exponents of quantum hypothesis testing~\cite{audenaert_2008,nagaoka_2006_converse_theorem,Hayashi_2007_error_exponent} and revisiting their extensions to settings with additional symmetries~\cite{Hiai_2008, Lipka_Bartosik_2024}, we obtain the following characterization.

\begin{thm}\label{thm: error/strong converse of work extraction}
The reliability function (error exponent) of work extraction from a quantum state $\rho\in\mD(\mH)$ with the target work extraction rate $r$ under the classes of Gibbs-preserving operations and thermal operations is characterized as
    \bal\label{eq: error exponent of work extraction_main}
    B_{{\rm GPO}}(\rho; r)&=\sup_{0<\alpha<1}\frac{\alpha-1}{\alpha}\qty(r-\pD_\alpha(\rho\|\tau))\\
    B_{{\rm TO}}(\rho; r)&=\sup_{0<\alpha<1}\frac{\alpha-1}{\alpha}\qty(r-\sD_\alpha(\rho\|\tau)),
    \eal
and in general it holds that $B_{{\rm GPO}}(\rho; r) > B_{{\rm TO}}(\rho; r)$ for rates $r < D(\rho\|\tau)$, unless $\rho$ commutes with the thermal state $\tau$.
\end{thm}
Previous studies of work extraction showed that the presence or absence of time-translation covariance does not affect the optimal asymptotic work extraction rate~\cite{Brandao_resource_theory_of_quantum, Gour_role_of_quantum_coherence}. 
This has led to the expectation that time-translation covariance is not essential in asymptotic work extraction. 
However, Eq.~\eqref{eq: error exponent of work extraction_main} in Theorem~\ref{thm: error/strong converse of work extraction} shows that this intuition breaks down at the level of precision.

Specifically, the presence or absence of time-translation covariance appears explicitly as the difference between the underlying divergences governing the error exponent. 
In particular, the achievable error exponent is strictly smaller under time-translation-covariant operations when $[\rho,\tau]=0$, which follows directly from the equality condition of the Araki--Lieb--Thirring inequality~\cite{araki_inequality_1990, LiebThirring} studied in Ref.~\cite{Friedland1994}

We can also find an example of the input state $\rho$ and the target work extraction rate $r>0$ such that $B_{{\rm GPO},{\rm ext}}(\qty{\rho^{\otimes n}}_{n\in\mbN}; r)$ diverges while $B_{{\rm TO},{\rm ext}}(\qty{\rho^{\otimes n}}_{n\in\mbN}; r)$ remains finite, which implies that one can achieve the work extraction from quantum states $\qty{\rho^{\otimes n}}_{n\in\mbN}$ of rate $r$ exactly under Gibbs-preserving operations, but not under thermal operations.
Further discussion of these examples is deferred to the Appendix.

\paragraph{Fundamental  bound on quality of work extraction.}\,---\,%
We now turn to an extreme regime of work extraction. 
So far, we have characterized the error when the work extraction rate is constrained to be no smaller than a fixed threshold $r$. 
Here, we instead consider the question of how reliably work extraction can be achieved if the exact rate is not of importance --- this can be understood as extracting an arbitrarily large but constant amount of work from $\rho^{\otimes n}$ asymptotically, and is conventionally termed the \emph{zero-rate} regime. A similar setting was recently studied in the context of distillation of quantum entanglement~\cite{lami_2024_asymptotic_quantification}. 
The zero-rate error exponent of work extraction is formally defined as
\bal
B^{\rm Z. R.}_{\mbO}(\rho):=\lim_{W\to\infty}\liminf_{n\to\infty}-\frac{1}{n}\log\mE_{\mbO}(\rho^{\otimes n}; \beta W)
\eal
The quantity can be evaluated by using the connection with quantum hypothesis testing established in Proposition~\ref{pro: one-shot error of work extraction}, which leads to the following. 

\begin{thm}\label{thm: zero-rate error}
    The zero rate error exponent of work extraction under the classes of Gibbs-preserving operations and thermal operations is characterized as
    \bal\label{eq: zero-rate error}
    B^{\rm Z. R.}_{{\rm GPO}}(\rho)=D(\tau\|\rho),~~
    B^{\rm Z. R.}_{{\rm TO}}(\rho)=D^\star(\tau\|\rho).
    \eal
    Here, $D^\star(\sigma\|\rho)$ is the star divergence defined as $D^\star(\sigma\|\rho):=\lim_{\alpha\to 1}\lim_{n\to\infty}\frac{1}{n}D_\alpha(\sigma^{\otimes n}\|\mP_{\sigma^{\otimes n}}(\rho^{\otimes n}))$~\cite{Audenaert_alpha_z_2015, Lipka_Bartosik_2024}.
\end{thm}

From the data-processing inequality of the Petz and sandwiched Rényi divergences, it follows that $D(\tau\|\rho)\geq D^\star(\tau\|\rho)$; the two are, however, different in general, and Theorem~\ref{thm: zero-rate error} shows that their difference exactly quantifying the extent of the operational gap between Gibbs-preserving and thermal operations in the zero-rate regime of reliable work extraction.

Moreover, Theorem~\ref{thm: zero-rate error} provides a new operational interpretation of the star divergence as the fundamental quantity governing the ultimate precision of zero-rate work extraction.
Interestingly, due to Ref.~\cite[Theorem 2]{Audenaert_alpha_z_2015}, $B^{\rm Z. R.}_{{\rm GPO}}(\rho)$ remains finite whenever $B^{\rm Z. R.}_{{\rm TO}}(\rho)$ is finite, and $B^{\rm Z. R.}_{{\rm TO}}(\rho)$ diverges whenever $B^{\rm Z. R.}_{{\rm GPO}}(\rho)$ diverges.

\paragraph{Data-processing inequality.}\,---\,%
So far, we have characterized the error exponents of work extraction. 
Remarkably, the sandwiched Rényi divergence of order $\alpha\in(0,\frac{1}{2})$ appearing in Eq.~\eqref{eq: error exponent of work extraction_main}, as well as the star divergence appearing in Eq.~\eqref{eq: zero-rate error}, do not satisfy the data-processing inequality for general channels~\cite{Berta_2017_on_variational}, which is generally regarded as a fundamental requirement for divergences to admit an operational interpretation. 
Despite this, we show that both quantities satisfy the data-processing inequality under an appropriate covariance condition on the channel.

\begin{pro}\label{pro: DPI of sandwiched}
    Let $\sigma\in\mD(\mH)$ be a quantum state and $\mP_\sigma$ be a pinching channel corresponding to $\sigma$.
    Furthermore, let $\Lambda:\mD(\mH_A)\to\mD(\mH_B)$ be a channel satisfying $\mP_{\Lambda(\sigma)^{\otimes n}}\circ \Lambda^{\otimes n}\circ\mP_{\sigma^{\otimes n}}=\mP_{\Lambda(\sigma)^{\otimes n}}\circ \Lambda^{\otimes n}, ~~\forall n\in\mbN$.
    Then, for any quantum state $\rho\in\mD(\mH)$, we have
    \bal
    \sD_\alpha(\rho\|\sigma)&\geq \sD_\alpha(\Lambda(\rho)\|\Lambda(\sigma)),~\forall \alpha \in [0,\infty],\\
    D^\star(\sigma\|\rho)&\geq D^\star(\Lambda(\sigma)\|\Lambda(\rho)).
    \eal
\end{pro}

If we take $\sigma$ to be the thermal state $\tau$ and $\Lambda$ to be a thermal operation, the above condition is satisfied. 
For any Rényi divergence $\mbD_\alpha$ satisfying the data-processing inequality, one can define a resource measure in quantum thermodynamics as $\mR_{\mbD_\alpha}(\rho):=\mbD_\alpha(\rho\|\tau)$.
Among additive Rényi-type divergences satisfying the data-processing inequality, the sandwiched Rényi divergence is known to be minimal. Thus, the corresponding resource measure is no larger than any other resource measure defined in this way. 
In particular, Proposition~\ref{pro: DPI of sandwiched} implies that the resource measure $\mR_{\sD_\alpha}$ is monotone under thermal operations for any $\alpha\in(0,\infty]$, even in the regime where the sandwiched Rényi divergence fails to satisfy the standard data-processing inequality under arbitrary quantum channels, and is in particular not monotone under general Gibbs-preserving operations.
Therefore, Theorem~\ref{thm: error/strong converse of work extraction} can be understood as an operational interpretation of the minimal Rényi-type resource measure in quantum thermodynamics.

\paragraph{Precision of general resource manipulation.}\,---\,%
Work extraction from a quantum state can be seen as a special case of resource distillation, one of the central tasks in the framework of resource theory.
In the Appendix, we further extend our analysis to the setting of general resource theories, which covers various kinds of quantum resources such as entanglement, speakable coherence, and magic.
There, we obtain the fundamental bounds on the precision of the resource distillation and the resource dilution within the general framework.

\paragraph{Discussion.}\,---\,%
We analyzed the asymptotic reliability of work extraction and characterized both the error exponent and the zero-rate error exponent under axiomatic and physically implementable classes of thermodynamic operations. 
Our results reveal a sharp operational separation between Gibbs-preserving operations and thermal operations: although the two classes achieve the same first-order asymptotic work extraction rate, they exhibit distinct achievable error exponents. 
This separation is completely characterized by the information-theoretic divergences appearing in the error decay.

More specifically, the separation at finite extraction rates is characterized by the distinction between the Petz and sandwiched Rényi divergences, while the zero-rate separation is characterized by the star divergence and the reversed Umegaki relative entropy. 
These results provide new operational interpretations of the sandwiched Rényi divergence for $\alpha\in(0,1)$ and of the star divergence from a thermodynamic perspective. 
Although neither the sandwiched Rényi divergence for $\alpha<1/2$ nor the star divergence satisfies the standard data-processing inequality under arbitrary quantum channels, we showed that both quantities satisfy monotonicity under time-translation-covariant channels, identifying symmetry-restricted monotonicity as the operational principle underlying their thermodynamic relevance.

Our results show that the distinction of the performance of work extraction between mathematically convenient and physically implementable thermodynamic operations emerges only beyond the first-order asymptotic regime, and manifests itself as a gap between fundamental information-theoretic quantities. 
This reveals new features of the fundamental connections between the precision of quantum resource manipulation and quantum hypothesis testing.

An important future direction is to extend our analysis to state-agnostic work extraction. 
In particular, previous studies showed that asymptotically optimal work extraction can be achieved without any prior information about the input state~\cite{Watanabe_universal,faist_2026_universal_thermodynamic}. 
Determining the optimal error exponent for such universal work extraction remains an important open problem.

Beyond quantum thermodynamics, our results suggest that characterizing the precision of resource manipulation provides a natural operational framework for the applications information-theoretic quantities in general quantum resource theories, with potential applications to refining the ultimate limits in the preparation of high-fidelity nonequilibrium states required in settings such as quantum error correction and quantum computation.

\emph{Note added.}---We became aware of an independent work by Munan Zhang and Kun Fang on quantum thermodynamics. Although the main setting and results of our two works are unrelated, the authors there obtained also a result on the reliability of work extraction under Gibbs-preserving operations, which appears concurrently in the updated manuscript~\cite{zhang2026quantumthermodynamicsuncertainequilibrium}.

%TC:ignore
\paragraph{Acknowledgments.---} 
The authors acknowledge Kun Fang for discussion.
This work is supported by the Japan Science and Technology Agency (JST) CREST Grant No.\ JPMJCR23I3, JST PRESTO Grant No.\ JPMJPR25FB , JSPS KAKENHI Grant Nos.\ 24K16975,  24K16984, 25K00924, 26H02015, 26KJ0965, and JST NEXUS Grant Number JPMJNX26C2, and the World-Leading Innovative Graduate Study Program for Advanced Basic Science Course at the University of Tokyo. MT is supported by the NRF Investigatorship award (NRF-NRFI10-2024-0006).

%TC:endignore

\let\oldaddcontentsline\addcontentsline% Store \addcontentsline
\renewcommand{\addcontentsline}[3]{}% Make \addcontentsline a no-op

\bibliographystyle{apsrmp4-2}
\bibliography{myref}

\let\addcontentsline\oldaddcontentsline% Restore \addcontentsline

%%%%%%%%%%%%%%%%%%%%%%%%%%%%%%%%%%%%%%%%%%
%%%%%%%%%%%%%%%%%%%%%%%%%%%%%%%%%%%%%%%%%%
%%%%%%%%%%%%%%%%%%%%%%%%%%%%%%%%%%%%%%%%%%
%%%%%%%%%%%%%%%%%%%%%%%%%%%%%%%%%%%%%%%%%%

\clearpage
\newgeometry{hmargin=1.2in,vmargin=0.8in}

\widetext

\appendix

\setcounter{thm}{0}
\renewcommand{\thethm}{S.\arabic{thm}}
\setcounter{figure}{0}
\renewcommand{\thefigure}{S.\arabic{figure}}

\begin{center}
{\large \bf Appendices}
\end{center}

\tableofcontents

\section{Introduction}
Understanding how well one can exploit the advantages of the quantum features is a central problem of quantum information theory.
Specifically, recent works have found that various properties of quantum states called \emph{quantum resources} empower quantum information processing.
Recent progress in the framework of quantum resource theories~\cite{Chitamber_Gour, Gour_2025_book} enables us to characterize fundamental limitations of various kinds of resources---the features of quantum states and quantum processes which cannot be increased in the given operational setting such as entanglement~\cite{Horodecki_4_2009}, nonstabilizerness~\cite{Veitch_2014_resource_theory, Howard_application_of}, athermality~\cite{horodecki_2013, Gour_role_of_quantum_coherence}, coherence~\cite{aberg2006quantifyingsuperposition, Baumgratz_2014_quantifying,Streltsov_2017_Colloquium}, and asymmetry~\cite{Bartlett_2007_reference, Gour_2008_resource_theory_of}. In particular, the framework of general resource theories provides a comprehensive perspective on the important features shared by various resources.

Previous papers studying the general resource theory have characterized the optimal performance of the state conversion with respect to the \emph{rate} $R$ of the conversion, that is, they study whether one can achieve the conversion $\rho^{\otimes n}\to\approx_{\ve_n}\sigma^{\otimes \lfloor Rn\rfloor}$ is possible with a permissible error, which has to vanish in the asymptotic limit $n\to\infty$.
Here, an important observation is that these studies do not specify the rate at which the error decays. This implies that, if one would like to convert quantum states to other states at a nearly optimal conversion rate, the error on the conversion is sufficiently small only if the number of copies of the initial state is sufficiently large.
If one consumes finite copies of the initial state, final states are not always close to the target states.
From the complementary perspective to the conventional studies of resource manipulation, it is fundamentally and practically vital to understand the fundamental limit on the \emph{precision} of resource manipulation.

Some recent papers have examined the exponent $c$ of the fidelity error $\ve_n\sim 2^{-cn}$ with which one can distill Bell states $\ket{\Phi}:=\frac{1}{\sqrt{2}}(\ket{0}\ket{0}+\ket{1}\ket{1})$ from noisy entangled states~\cite{Hayashi_2002_error_exponents,Hayashi_2006_general_formula,lami_2024_asymptotic_quantification, lin_2026_exponential_analysis}. 
Furthermore, in Ref.~\cite{Wang_2019, berta_2025_strong_converse}, they study the strong converse exponent, the speed with which the error accumulates, in the quantum dichotomies setting which connects the setting of quantum thermodynamics.

In this work, we further extend the previous observations to a far broader range of tasks. We study the error exponent and the strong converse exponent of two main resource manipulation tasks called resource distillation and resource dilution, employing the framework of general resource theories. This result extends the result of Ref.~\cite{lin_2026_exponential_analysis} in a more general case.
Moreover, as a byproduct, we derive the strong converse exponent of the smoothed max relative entropy of sets of states. This also extends the previous observations in Ref.~\cite{Wang_2019, fang_2025_error_exponents}.

We further employ the observations in a specific resource theory, namely the resource theory of thermodynamics, where the nonequilibriumness serves as an important resource for quantum computation. 
An important property of quantum thermodynamics is that the available operations must obey energy conservation, which leads to time-translation covariance.
However, a major observation is that, as long as we consider work extraction, a central task in quantum thermodynamics, whether the available operations are restricted by time-translation covariance does not affect the amount of extractable work in the thermodynamic limit~\cite{Brandao_resource_theory_of_quantum, Wang_2019, Gour_role_of_quantum_coherence}.
This leads to the belief that time-translation invariance does not play a central role in work extraction in the thermodynamic limit.

In this work, we characterize the error exponent of work extraction and clarify that the presence of the time-translation covariance indeed restricts the performance of work extraction by focusing on the speed of the error decay, even in the asymptotic limit. Interestingly, the difference in the error exponent of work extraction under two classes of operations with and without time-translation covariance can be clearly represented as the different Rényi relative entropy appearing in the expression.

\section{Preliminaries}
Throughout this work, we consider a finite-dimensional Hilbert space. We denote the set of linear maps on a Hilbert space $\mH$ as $\mL(\mH)$.
Let $\mD(\mH)=\qty{\rho\in\mL(\mH)~|~\rho\geq 0, \Tr\rho=1}, ~\mD_\leq (\mH)=\qty{\rho\in\mL(\mH)~|~\rho\geq 0, \Tr\rho\leq 1}, ~\mP(\mH)=\qty{X\in\mL(\mH)~|~X\geq 0}$ denote the set of quantum states, subnormalized quantum states, and positive semidefinite matrices on a Hilbert space $\mH$, respectively. Moreover, we denote the set of quantum channels (completely positive trace-preserving map, or CPTP map in short) from $\mD(\mH_A)$ to $\mD(\mH_B)$ as ${\rm CPTP(\mH_A\to \mH_B)}$.

\subsection{Distance measures and divergences}

We first review distance measures and information-theoretic quantities used in this work.
Trace distance $T(\rho,\sigma)$ between two states $\rho, \sigma\in\mD(\mH)$ is defined as $T(\rho,\sigma):=\frac{1}{2}\|\rho-\sigma\|_1$. 
The (square) fidelity $F(\rho,\sigma)$ between two states is defined as $F(\rho,\sigma):=\| \sqrt{\rho}\sqrt{\sigma} \|_1^2$.
The purified distance $P(\rho,\sigma)$ and the infidelity $I(\rho,\sigma) $are defined using the fidelity as $P(\rho,\sigma)=:\sqrt{1-F(\rho,\sigma)}$, $I(\rho,\sigma):=P^2(\rho,\sigma)$.
The Petz $\alpha$-Rényi relative entropy $\pD_\alpha(\rho\|\sigma)$ and the sandwiched $\alpha$-Rényi relative entropy $\sD_\alpha(\rho\|\sigma)$ are defined for $\rho\in\mD_{\leq}(\mH)$ and $\sigma\in \mP(\mH)$ as 
\bal
\pD_\alpha(\rho\|\sigma)=
\begin{cases}
    \frac{1}{\alpha-1}\log\Tr[\rho^\alpha\sigma^{1-\alpha}]~~&\mbox{ if }(\alpha<1\land \rho\not\perp\sigma)\lor\supp(\rho)\subset\supp(\sigma),\\
    +\infty~~&\mbox{otherwise},
\end{cases}
\eal
\bal
\sD_\alpha(\rho\|\sigma)=
\begin{cases}
    \frac{1}{\alpha-1}\log\Tr\qty[\qty(\sigma^{\frac{1-\alpha}{2\alpha}} \rho\sigma^{\frac{1-\alpha}{2\alpha}} )^\alpha]&\mbox{ if }(\alpha<1\land \rho\not\perp\sigma)\lor\supp(\rho)\subset\supp(\sigma),\\
    +\infty~~&\mbox{otherwise}.
\end{cases}
\eal
Umegaki relative entropy is defined as 
\bal
D(\rho\|\sigma):=\Tr[\rho\log \rho-\rho\log \sigma],
\eal
which is also obtained as the limit
\bal
\lim_{\alpha\to 1}\pD_\alpha(\rho\|\sigma)=\lim_{\alpha\to 1}\sD_\alpha(\rho\|\sigma)=D(\rho\|\sigma).
\eal
These relative entropies satisfy the \emph{data-processing inequality}, a fundamental property that implies one can discriminate between two states no better after undergoing the same physical process: For any $\rho,\sigma$ and for any $\Lambda \in {\rm CPTP}(\mH\to \mH')$, it holds that
\bal
\pD_\alpha(\rho\|\sigma)&\geq \pD_\alpha(\Lambda(\rho)\|\Lambda(\sigma))~~\forall \alpha\in[0,2],\\
\sD_\alpha(\rho\|\sigma)&\geq \sD_\alpha(\Lambda(\rho)\|\Lambda(\sigma))~~\forall \alpha\in\left[\frac{1}{2},\infty\right),\\
D(\rho\|\sigma)&\geq D(\Lambda(\rho)\|\Lambda(\sigma)).
\eal

Some limits of the sandwiched Rényi relative entropy and Petz Rényi relative entropy are important and often discussed throughout this paper.
The max relative entropy $D_{\max} (\rho\|\sigma)$ of $\rho\in\mD_{\leq}(\mH)$ with respect to $\sigma\in\mP(\mH)$ is defined as
\bal
D_{\max}(\rho\|\sigma):=\lim_{\alpha\to \infty}\sD_\alpha(\rho\|\sigma)=\log\inf\qty{\lambda~|~\rho\leq \lambda\sigma}.
\eal
The min relative entropy $D_{\min} (\rho\|\sigma)$ of $\rho\in\mD_{\leq}(\mH)$ with respect to $\sigma\in\mP(\mH)$ is defined as
\bal
D_{\min}(\rho\|\sigma)=-\log \Tr[\Pi_{\supp(\rho)}\sigma],
\eal
where $\Pi_{\supp(\rho)}:=\rho^0$ denotes the projector onto the support $\supp(\rho):=(\ker \rho)^\perp$ of $\rho$.

In this paper, we often make use of these relative entropies defined not only for quantum states but also for sets of quantum states. For a relative entropy $\mbD(\cdot\|\cdot)$ and nonempty sets $\mA,\mB\subset\mD(\mH)$ of quantum states, the corresponding quantity is defined as
\bal
\mbD(\mA\|\mB):=\inf_{\substack{\rho\in\mA \\ \sigma\in\mB}}\mbD(\rho\|\sigma).
\eal

\subsubsection{Pinching}
In this paper, we often make use of the properties of a channel called \emph{pinching}, a useful tool in quantum information theory~\cite{Hayashi_2002_optimal_sequence}.
Suppose that we are given a complete set of projectors $\qty{\Pi_i}_i$, satisfying $\sum_i\Pi_i=I$.
The pinching with respect to this set of projectors is defined as 
\bal
\mP_{\qty{\Pi_i}_i}(\cdot):=\sum_i\Pi_i(\cdot)\Pi_i.
\eal
The pinching channel satisfies the pinching inequality 
\bal\label{eq: pinching inequality}
\mP_{\qty{\Pi_i}_i}(\rho)\geq \frac{\rho}{\abs{\qty{\Pi_i}_i}}, ~\forall\rho.
\eal
We especially make use of the pinching channel with respect to the spectral decomposition of a quantum state $\sigma=\sum_i\lambda_i\Pi_i$, where each $\lambda_i$ is a distinct eigenvalue, defined as 
\bal\label{eq: spectral pinching}
\mP_{\sigma}(\cdot):=\mP_{\qty{\Pi_i}_i}(\cdot).
\eal
Due to the data-processing inequality of the sandwiched Rényi relative entropy with respect to pinching and the pinching inequality, we have~\cite[Lemma 3]{Hayashi_2016}
\bal
\sD_\alpha(\mP_\sigma(\rho)\|\sigma)\leq \sD_\alpha(\rho\|\sigma)\leq \sD_\alpha(\mP_\sigma(\rho)\|\sigma)+2\log\abs{\spec(\sigma)},~~\forall\alpha\geq 0.
\eal
Noting that the number $\abs{\spec(\sigma^{\otimes n})}$ of distinct eigenvalues of an i.i.d. state $\sigma^{\otimes n}$ is at most polynomial~\cite{cover_1999_elements}, it holds that~\cite[Proposition 4.12]{Tomamichel_2016}
\bal\label{eq: limit of pinched sandwich}
\lim_{n\to\infty}\frac{1}{n}\sD_\alpha(\mP_{\sigma^{\otimes n}}(\rho^{\otimes n})\|\sigma^{\otimes n})=\sD_\alpha(\rho\|\sigma).
\eal

\subsection{Quantum hypothesis testing}
We now review standard notions in the information-theoretic task of quantum hypothesis testing.
Hypothesis testing is the task of discriminating between two hypotheses, namely \emph{null hypothesis} $\rho\in\mD(\mH)$ and \emph{alternative hypothesis} $\sigma\in\mD(\mH)$. Suppose that we are given either of the hypotheses, $\rho$ and $\sigma$, and our goal is to guess which state is actually given by performing a test described by POVM $\qty{M, I-M}$. Here, $M$ represents the measurement outcome corresponding to the guess that the given state is $\rho$, and $I-M$ represents the guess that the given state is $\sigma$.
For a test $\qty{M, I-M}$, the probability of type I error $\alpha(\rho, M)$ and type II error $\beta(\sigma, M)$ is defined as
\bal
\alpha(\rho,M):=\Tr[\rho (I-M)], ~~\beta(\sigma,M):=\Tr[\sigma M].
\eal
Throughout this manuscript, we consider the asymmetric setting where we would like to minimize the probability of type II error while the probability of type I error is kept no larger than a constant $\ve>0$. The hypothesis testing divergence, the figure of merit of the asymmetric hypothesis testing, is defined as
\bal
D^\ve_H(\rho\|\sigma)=-\log \inf_{\substack{0\leq M\leq I\\ \Tr[\rho(I-M)]\leq \ve }}\Tr[\sigma M].
\eal

When multiple copies of $\rho$ or $\sigma$ are given, the discrimination task becomes easier, and the probability of these errors decays fast.
A series of previous papers has characterized the optimal exponent at which the error vanishes.

We first consider the limit of the hypothesis testing divergence called Stein's exponent, which is known to be characterized by the Umegaki relative entropy as follows~\cite{hiai_1991_proper, Ogawa_2000_strong}:
\bal
\lim_{n\to\infty}\frac{1}{n}D^\ve_H(\rho^{\otimes n}\|\sigma^{\otimes n})=D(\rho\|\sigma),~~\forall\ve\in(0,1).
\eal
We can also consider other scenarios where the probability of type II error decays faster (resp. slower) than Stein's exponent, and study how fast the probability of type I error decays to 0 (resp. accumulates to 1), which are precisely the quantities called the error exponent and the strong converse exponent.

The error exponent and the strong converse exponent of the hypothesis testing between the i.i.d. sequences $\qty{\rho^{\otimes n}}_{n\in\mbN}$ and $\qty{\sigma^{\otimes n}}_{n\in\mbN}$ are defined as follows: 
\bal
B_H(\qty{\rho^{\otimes n}}_{n\in\mbN}\|\qty{\sigma^{\otimes n}}_{n\in\mbN}~;~r)&:=\sup_{\qty{M_n}_n}\qty{\liminf_{n\to\infty}-\frac{1}{n}\log\alpha(\rho^{\otimes n}, M_n)~|~\liminf_{n\to\infty}-\frac{1}{n}\log\beta(\sigma^{\otimes n},M_n)\geq r}\\
B_H^*(\qty{\rho^{\otimes n}}_{n\in\mbN}\|\qty{\sigma^{\otimes n}}_{n\in\mbN}~;~r)&:=\inf_{\qty{M_n}_n}\qty{\limsup_{n\to\infty}-\frac{1}{n}\log(1-\alpha(\rho^{\otimes n}, M_n))~|~\liminf_{n\to\infty}-\frac{1}{n}\log\beta(\sigma^{\otimes n},M_n)\geq r}
\eal
The exact expressions for these quantities, due respectively to~\cite{audenaert_2008,nagaoka_2006_converse_theorem,Hayashi_2007_error_exponent} and~\cite{Mosonyi_Ogawa_2014}, are 
\bal\label{eq:exponents_hyp_testing}
B_H(\qty{\rho^{\otimes n}}_{n\in\mbN}\|\qty{\sigma^{\otimes n}}_{n\in\mbN}~;~r)&=\sup_{0<\alpha<1}\frac{\alpha-1}{\alpha}\qty(r-\pD_\alpha(\rho\|\sigma)),\\
B_H^*(\qty{\rho^{\otimes n}}_{n\in\mbN}\|\qty{\sigma^{\otimes n}}_{n\in\mbN}~;~r)&=\sup_{\alpha>1}\frac{\alpha-1}{\alpha}\qty(r-\sD_\alpha(\rho\|\sigma)).
\eal

We can also consider a generalization of the quantum hypothesis testing known as the composite quantum hypothesis testing, where the null hypothesis and the alternative hypothesis are given as the sequences of sets $\mA=\qty{\mA_n}_{n\in\mbN}, \mB=\qty{\mB_n}_{n\in\mbN}$.
Our goal is to distinguish between the states in either of these sets by binary measurement corresponding to the POVM elements $\qty{M, I-M}$.

The probabilities of type I error $\alpha(\mA_n,M_n)$ and type II error $\beta(\mB_n,M_n)$ are defined as
\bal
\alpha(\mA_n,M_n)&=\sup_{\rho_n\in\mA_n}\Tr[\rho_n(I-M_n)], \\
\beta(\mB_n,M_n) &=\sup_{\sigma_n\in\mB_n}\Tr[\sigma_n M_n].
\eal
The hypothesis testing divergence between these two hypotheses is defined as 
\bal
D^\ve_H(\mA_n\|\mB_n)=-\log\inf_{\substack{0\leq M_n\leq I \\ \alpha(\mA_n,M_n)\leq \ve}} \beta(\mB_n,M_n)
\eal

The error exponent and the strong converse exponent of the composite hypothesis testing are defined as
\bal
B_H(\mA\|\mB ;r)&:=\sup_{\qty{M_n}_n}\qty{\liminf_{n\to\infty}-\frac{1}{n}\log\alpha(\mA_n,M_n)~|~\liminf_{n\to\infty}-\frac{1}{n}\log \beta(\mB_n,M_n)\geq r}\\
&=\sup\qty{\liminf_{n\to\infty}-\frac{1}{n}\log\ve_n~|~\liminf_{n\to\infty}\frac{1}{n}D^{\ve_n}_H(\mA_n\|\mB_n)\geq r}\\
B^*_H(\mA\|\mB ;r)&:=\inf_{\qty{M_n}_n}\qty{\limsup_{n\to\infty}-\frac{1}{n}\log \qty(1-\alpha(\mA_n,M_n))~|~\liminf_{n\to\infty}-\frac{1}{n}\log \beta(\mB_n,M_n)\geq r}\\
&=\inf\qty{\limsup_{n\to\infty}-\frac{1}{n}\log(1-\ve_n)~|~\liminf_{n\to\infty}\frac{1}{n}D^{\ve_n}_H(\mA_n\|\mB_n)\geq r}
\eal

Suppose that the hypotheses $\mA=\qty{\mA_n}_{n\in\mbN}, \mB=\qty{\mB_n}_{n\in\mbN}$ satisfy the following:
\begin{enumerate}
    \item For any $n\in\mbN$, $\mA_n, \mB_n$ are convex and compact.
    \item The sequence $\mA=\qty{\mA_n}_{n\in\mbN}, \mB=\qty{\mB_n}_{n\in\mbN}$ are closed under taking tensor-product, that is, $\mA$ and $\mB$ satisfy $\mA_n\otimes \mA_m\subset\mA_{m+n}$ and $\mB_n\otimes \mB_m\subset\mB_{m+n}$.
\end{enumerate}
In Ref.~\cite{fang_2025_error_exponents}, it is shown that, for the composite hypotheses satisfying the conditions above, it holds that 
\bal\label{eq: composite Hopeffding}
B_H(\mA\|\mB ;r)&=\lim_{n\to\infty}\frac{1}{n}\inf_{\substack{\rho_n\in\mA_n\\ \sigma_n\in\mB_n}} \sup_{\alpha\in (0,1)}\frac{\alpha-1}{\alpha}\qty(nr-\pD_{\alpha}(\rho_n\|\sigma_n)),\\
B^*_H(\mA\|\mB ;r)&\geq\lim_{n\to\infty}\frac{1}{n}\inf_{\substack{\rho_n\in\mA_n\\ \sigma_n\in\mB_n}} \sup_{\alpha>1}\frac{\alpha-1}{\alpha}\qty(nr-\sD_{\alpha}(\rho_n\|\sigma_n)).
\eal
There is however a considerable difficulty in interchanging the limit $n\to\infty$ and the optimizations that appear in these formulas~\cite{berta_2023,fang_2025_error_exponents}, preventing a simplification of the expressions and in many cases hindering their applications.  Remarkably, recent developments in composite hypothesis testing show that the somewhat simpler question of the asymptotics of the hypothesis testing for a constant error $\ve$ can indeed be obtained: under suitable assumptions on $\mA_n$ and $\mB_n$, it is given exactly by the regularized relative entropy between the sets, namely~\cite{hayashi_generalized_2025,Lami_2025_gqsl,lami_2024_asymptotic_quantification,hayashi_2025-5,fang_2025_generalized_quantum_AEP,lami_2025-2}
\bal
\lim_{n\to\infty}\frac{1}{n}D^\ve_H(\mA_n\|\mB_n)= \lim_{n\to\infty} \frac1n \min_{\substack{\rho_n\in\mA_n\\ \sigma_n\in\mB_n}} D(\rho_n\|\sigma_n),~~\forall\ve\in(0,1).
\eal
Whether these results can be strengthened to an asymptotic analysis of error exponents, particularly in a way that would help characterize quantum resource manipulation, is an exciting open question in quantum hypothesis testing.

\subsection{Smoothed max relative entropy}
We review another important quantity employed throughout this paper, namely the smoothed max relative entropy. 
To define this quantity, we need to specify which set of matrices and which distance measure we use for the smoothing.
For this purpose, we define $\ve$-balls around a quantum state $\rho$ with respect to a distance measure $d(\cdot,\cdot)$ for quantum states and subnormalized quantum states as 
\bal
\mB^\ve_{d,=}(\rho)&:=\qty{\tilde{\rho}\in\mD(\mH)~|~d(\rho,\tilde{\rho})\leq \ve},\\
\mB^\ve_{d,\leq}(\rho)&:=\qty{\tilde{\rho}\in\mD_{\leq}(\mH)~|~d(\rho,\tilde{\rho})\leq \ve}.
\eal

Smoothed max relative entropy is defined by taking the optimization over either $\mB^\ve_{d,=}(\rho)$ or $\mB^\ve_{d,\leq }(\rho)$ as 
\bal
D^{\ve, d, (*)}_{\rm max}(\rho\|\sigma):=\min_{\tilde{\rho}\in\mB^\ve_{d,(*)}(\rho)}\log \min\qty{\lambda~|~\tilde{\rho}\leq \lambda\sigma},~~(*)\in\qty{=,~\leq}.
\eal
From the definition, it holds that $D^{\ve, d, =}_{\rm max}(\rho\|\sigma)\geq D^{\ve, d, \leq}_{\rm max}(\rho\|\sigma)$.
In Ref.~\cite{Regula_2026_tight_relation}, it is shown that for a trace distance $T(\rho,\sigma)$ and $P(\rho,\sigma)$, these two definitions are connected as
\bal
D^{\ve, d,=}_{\rm max}(\rho\|\sigma)=\max\qty{ D^{\ve, d, \leq}_{\rm max}(\rho\|\sigma),0}.
\eal
When we take the distance measure $d$ as the purified distance $P(\rho,\sigma)$, it is shown in Ref.~\cite{Tomamichel_2009_Quantum_AEP} that
\bal
\lim_{n\to\infty}\frac{1}{n}D^{\ve,P,\leq}_{\rm max}(\rho^{\otimes n}\|\sigma^{\otimes n})=D(\rho\|\sigma), ~~\forall\ve\in(0,1)
\eal
holds for any $\rho,\sigma\in\mD(\mH)$.

We can consider the error exponent and the strong converse exponent for the smoothed max relative entropy, which is defined for a smoothing with respect to $d, (*)\in\qty{=,~\leq }$ as 
\bal
B^{d,(*)}_M(\qty{\rho^{\otimes n}}_{n\in\mbN}\|\qty{\sigma^{\otimes n}}_{n\in\mbN} |r)&:=\sup_{\qty{\ve_n}}\qty{\liminf_{n\to\infty}-\frac{1}{n}\log\ve_n~|~\liminf_{n\to\infty}\frac{1}{n}D^{\ve_n, d,(*)}_{\rm max}(\rho^{\otimes n}\|\sigma^{\otimes n})\leq r}\\
B^{*,d,(*)}_M(\qty{\rho^{\otimes n}}_{n\in\mbN}\|\qty{\sigma^{\otimes n}}_{n\in\mbN} |r)&:=\inf_{\qty{\ve_n}}\qty{\limsup_{n\to\infty}-\frac{1}{n}\log(1-\ve_n)~|~\liminf_{n\to\infty}\frac{1}{n}D^{\ve_n, d, (*)}_{\rm max}(\rho^{\otimes n}\|\sigma^{\otimes n})\leq r}.
\eal

The previous research has found that the error exponent and the strong converse exponent of the smoothed max relative entropy are fully characterized as~\cite{Li_2023_tight_exponential,li_operational_2024}
\bal
B^{P,\leq }_M(\qty{\rho^{\otimes n}}_{n\in\mbN}\|\qty{\sigma^{\otimes n}}_{n\in\mbN} |r)&=\frac{1}{2}\sup_{\alpha>1}(\alpha-1)(r-\sD_\alpha(\rho\|\sigma)),\\
B^{*,P,\leq }_M(\qty{\rho^{\otimes n}}_{n\in\mbN}\|\qty{\sigma^{\otimes n}}_{n\in\mbN} |r)&=\sup_{\frac{1}{2}<\alpha<1}\frac{\alpha-1}{\alpha}(r-\sD_\alpha(\rho\|\sigma)).
\eal

We can also generalize the smoothed max relative entropy to accept sets of states as arguments, which is defined as
\bal
D^{\ve, d,(\*)}_{\rm max}(\mA\|\mB):=\inf_{\substack{\rho\in\mA\\ \sigma\in\mB}}D^{\ve, d,(\*)}_{\rm max}(\rho\|\sigma).
\eal
The notion of error exponents and the strong converse exponents can be generalized to the smoothed max relative entropy of sets of states as follows.
\bal
B^{d,(*)}_M(\qty{\mA_n}_{n\in\mbN}\|\qty{\mB_n}_{n\in\mbN} |r)&:=\sup_{\qty{\ve_n}}\qty{\liminf_{n\to\infty}-\frac{1}{n}\log\ve_n~|~\liminf_{n\to\infty}\frac{1}{n}D^{\ve_n, d,(*)}_{\rm max}(\mA_n\|\mB_n)\leq r}\\
B^{*,d,(*)}_M(\qty{\mA_n}_{n\in\mbN}\|\qty{\mB_n}_{n\in\mbN} |r)&:=\inf_{\qty{\ve_n}}\qty{\limsup_{n\to\infty}-\frac{1}{n}\log(1-\ve_n)~|~\liminf_{n\to\infty}\frac{1}{n}D^{\ve_n, d, (*)}_{\rm max}(\mA_n\|\mB_n)\leq r}.
\label{eq:smoothed max error exponents def}
\eal

\subsection{Quantum resource theories}
A quantum resource theory is identified by a subset $\mbF(\mH)\subset \mD(\mH)$ of states and a subset $\mbO(\mH_A\to\mH_B)\subset{\rm CPTP(A\to B)}$ of quantum channels such that for any $\Lambda\in\mbO(\mH_A\to\mH_B)$, $\Lambda(\mbF(\mH_A))\subset\mbF(\mH_B)$. 
The subset $\mbF(\mH)$ is called the set of \emph{free states}, the states which can be prepared without any cost in the given scenario.
Similarly, $\mbO(\mH_A\to\mH_B)$ is referred to as the class of \emph{free operations}, which is a collection of the operations applied freely in the given scenario.
The condition mentioned above for inclusion is the minimum condition that guarantees the resource cannot be generated from scratch by applying any free operations to free states.
Throughout this work, we assume that $\mbF(\mH)$ is convex and compact.

When we fix the set $\mbF(\mH)$of free states, the maximum choice of the set of free operations is the set of channels that includes all channels that map the set of free states to itself. This choice of free operations is called the set $\mbO_{\rm RNG}$ of \emph{resource non-generating operations}.
For instance, in the entanglement theory, this class of operations corresponds to the set of separability-preserving operations.
Unless stated otherwise, we consider the set of resource non-generating operations to be the set of free operations.

One of the most fundamental problems in the study of quantum resource theories is to determine whether the conversion from a state $\rho$ to another state $\sigma$ is possible under the restriction of the free operations with the permissible error $\ve>0$. 
This problem can be extended to the asymptotic scenario where we are given many copies of $\rho$, and one aims at obtaining as many copies of the target state $\sigma$ as possible, while the error on the final state must vanish in the asymptotic limit.

Among these conversions, the resource distillation and the resource dilution are the most important conversions.
Suppose that there is a specific family $\qty{\phi_d}_{d\in\mbD}$ of reference resource states, which depends on the settings of the given resource theories.
Here, $\mbD\subset\mbN$ denotes the valid dimensions for the reference resource states.
An example is the family of Bell states $\qty{\ketbra{\Phi_d}{\Phi_d}^{AB}}_{d\in\mbN}$, where $\ket{\Phi_d}^{AB}$ is defined as $\ket{\Phi_d}^{AB}=\frac{1}{\sqrt{d}}\sum_{i=1}^d\ket{i}^A\ket{i}^B$.
In this example, the set of valid dimensions is $\mbD=\qty{d^2~|~d\in\mbN}$.

Finding quantifiers, known as \emph{resource measures}, of the resourcefulness of quantum states is a central problem in quantum resource theory.
One canonical way to define the resource measure is to employ the relative entropies satisfying the data-processing inequality.
Another standard resource measure often employed in the literature is called \ emph {robustness}. Specifically, we make use of two variants of robustness, namely the standard and the generalized robustness defined as~\cite{vidal_1999,Brandao_2010_gqsl}
\bal
R_s(\rho):&=\inf\qty{\lambda\geq 0~|~\exists \tau\in\mbF \mbox{ s.t. } {\rho+\lambda\tau}\in\cone(\mbF)},\\
R_g(\rho):&=\inf\qty{\lambda\geq 0~|~\exists \tau\in\mD(\mH) \mbox{ s.t. }{\rho+\lambda\tau}\in\cone(\mbF)}.
\eal
From the definition we can easily see that $R_g(\rho)\leq R_s(\rho).$
Also, we can define the logarithmic standard robustness and the logarithmic generalized robustness as 
\bal
LR_s(\rho):=\log(1+R_s(\rho)),\\
LR_g(\rho):=\log(1+R_g(\rho)).
\eal
Here, it is easy to see that $LR_g(\rho)$ is equal to $D_{\rm max}(\rho\|\mbF)\coloneqq \min_{\sigma\in\mbF}D_{\rm max}(\rho\|\sigma)$. Similarly to the smoothed max relative entropy, we can also define the smoothed version of the logarithmic standard robustness as 
\bal
LR^{\ve, d, (*)}_s(\rho):=\min_{\tilde{\rho}\in\mB^\ve_{d,(*)}(\rho)}LR_s(\rho).
\eal

 \subsubsection{Resource distillation}

Resource distillation is the task of extracting the resource state with the largest possible dimension from given copies of the initial states.
The optimal performance of the resource distillation under the class $\mbO$ of free operations is represented by the maximum dimension of the reference resource state that can be obtained by applying a free operation to the input.
The one-shot version and the asymptotic version of the distillable resource are defined as follows.
\begin{defnboxed}
    \begin{defn}[Distillable resource]
    Let $\qty{\phi_d}_{d\in\mbD}$ be the family of reference pure resource states of the resource theory under consideration.
    $\ve$-one shot distillable resource from a quantum state $\rho$ under the class $\mbO$ of free operations with respect to a resource measure $\mfR$ is defined as
    \bal
    d^\ve_{\mbO,\mfR}(\rho):=\max\qty{\mfR(\phi_d)~|~\max_{\Lambda\in\mbO}F(\Lambda(\rho), \phi_d)\geq 1-\ve}.
    \eal
        Asymptotic distillable resource from a sequence $\qty{\rho^{\otimes n}}_{n\in\mbN}$ of i.i.d. states under the class $\mbO$ of free operations with respect to a resource measure $\mfR$ is defined as 
    \bal
    d^\infty_{\mbO,\mfR}(\rho):=\lim_{\ve\to +0}\limsup_{n\to\infty}\frac{1}{n}d^\ve_{\mbO,\mfR}(\rho^{\otimes n})
    \eal
    \end{defn}
\end{defnboxed}

On the other hand, resource dilution is the inverse of resource distillation, aiming to obtain as many copies of target states by applying free operations to a few pure reference resource states.
The optimal performance of resource dilution under the class $\mbO$ of free operations is quantified by the minimum dimension of the reference resource state which costs to obtain the target states.
Just as the distillable resource, the one-shot and asymptotic versions of the resource cost are defined as follows.

\begin{defnboxed}
    \begin{defn}[Resource cost]
    Let $\qty{\phi_d}_{d\in\mbD}$ be the family of reference pure resource states of the resource theory under consideration.
    $\ve$-one shot resource cost to obtain a quantum state $\rho$ under the class $\mbO$ of free operations with respect to a resource measure $\mfR$ is defined as 
    \bal
    c^\ve_{\mbO,\mfR}(\rho):=\min\qty{\mfR(\phi_d)~|~\max_{\Lambda\in\mbO}F(\Lambda(\phi_d), \rho)\geq 1-\ve}.
    \eal
        Asymptotic resource cost of a sequence $\qty{\rho^{\otimes n}}_{n\in\mbN}$ of i.i.d. states under the class $\mbO$ of free operations with respect to a resource measure $\mfR$ is defined as 
        \bal
        c^\infty_{\mbO,\mfR}(\rho):=\lim_{\ve\to +0}\limsup_{n\to\infty}\frac{1}{n}c^\ve_{\mbO,\mfR}(\rho^{\otimes n}).
        \eal
    \end{defn}
\end{defnboxed}

The previous studies along this line have revealed that the optimal resource distillation and the resource dilution are quantitatively connected to the information-theoretic quantities such as the hypothesis testing divergence and the robustness of resource~\cite{Liu_one_shot, Regula_Takagi_2021}.

\section{Error exponent and strong converse exponent of general resource distillation}
\subsection{One-shot optimal error of general resource distillation}
We are now in a position to discuss the main results. We study the smallest error with which one can extract the resource state from a fixed reference resource state $\phi_d$ from the initial state $\rho$.
The figure of merit in this scenario is defined as follows.
\begin{defnboxed}
    \begin{defn}
Let $(\mbF,\mbO)$ be the resource theory with the family $\qty{\phi_d}_{d\in\mbD}$ of the reference resource states. The optimal one-shot error of resource distillation from a quantum state $\rho$ with the target state $\phi_d$ is defined as 
\bal
\mE_{\mbO, {\rm distill}}(\rho;d)&:=1-\max_{\Lambda\in\mbO} F(\Lambda(\rho),\phi_d)
\eal
\end{defn}
\end{defnboxed}
Now, we connect these quantities with the information-theoretic quantities. When one takes an appropriate choice of the family of reference resource states, the one-shot optimal error of resource distillation is shown to be characterized through the hypothesis testing divergence.
\begin{proboxed}
    \begin{pro}\label{pro: one shot distillation}
        Let $(\mbF,\mbO={\rm RNG})$ be a quantum resource theory under resource non-generating operations. Assume that the family $\qty{\phi_d}_{d\in\mbD}$ of the reference states satisfy $D_{\rm min}(\phi_d\|\mbF)=LR_s(\phi_d)$.
        Then, the one-shot optimal error of resource distillation for a target state $\phi_d$ is characterized as 
        \bal
        \mE_{ {\rm RNG,distill}}(\rho;d)=\min\qty{\ve\geq 0~|~D_{\rm min}(\phi_d\|\mbF)\leq D^\ve_H(\rho\|\mbF)}.
        \eal
        When the family $\qty{\phi_d}_{d\in\mbD}$ of the reference states satisfy $D_{\rm min}(\phi_d\|\mbF)=D_{\rm max}(\phi_d\|\mbF)$ and $\Tr[\phi_d\sigma]=2^{-D_{\rm min}(\phi_d\|\mbF)}={{\rm const.}}$ for an arbitrary free state $\sigma\in\mbF$, the one-shot optimal error of resource distillation for a target state $\phi_d$ is characterized as 
        \bal
        \mE_{ {\rm RNG,distill}}(\rho;d)=\min\qty{\ve\geq 0~|~D_{\rm min}(\phi_d\|\mbF)\leq D^\ve_H(\rho\|\aff(\mbF))},
        \eal
        where $\aff(\mbF)=\qty{X\in{\mL(\mH)} ~|~ \exists t\in\mbR, \sigma_1,\sigma_2\in\mbF {~\rm{ s.t. }~} X=t\sigma_1+(1-t)\sigma_2}$ denotes the affine hull of $\mbF\in\mD(\mH)$.
    \end{pro}
\end{proboxed}
\begin{proof}
The proof follows directly from Ref.~\cite[Corollary 6,7]{Regula_benchmarking} as follows.
We first discuss the first part.
In Ref.~\cite[Corollary 6]{Regula_benchmarking}, it is shown that
\bal
1-\mE_{{\rm RNG}, {\rm distill}}(\rho;d)=G\qty(\rho, 2^{D_{\rm min}(\phi_d\|\mbF)}),
\eal
where $G(\rho,k)$ is defined for $k\geq 1$ as 
\bal
G(\rho,k):=\sup\qty{\Tr[\rho W]~|~0\leq W\leq I, ~W\in\frac{1}{k}\mbF^\circ},
\eal
and $\mbF^\circ$ is a polar set of $\mbF$ defined as 
\bal
\mbF^\circ:=\qty{X\in\mL(\mH)~|~\Tr[X\sigma]\leq 1,~\forall\sigma\in\mbF}.
\eal
From this, we can see that
\bal
1-\mE_{{\rm RNG}, {\rm distill}}(\rho;d)&=\sup\qty{\Tr[\rho W]~|~0\leq W\leq I,~\Tr[\sigma W]\leq 2^{-D_{\rm min}(\phi_d\|\mbF)},~\forall\sigma\in\mbF}\\
&=\sup\qty{1-\ve~|~D_{\rm min}(\phi_d\|\mbF)\leq D^\ve_H(\rho\|\mbF)}
\eal
From this, we reach the first claim.

Now, we discuss the second part.
Due to Ref.~\cite[Corollary 7]{Regula_benchmarking}, the one-shot optimal error of resource distillation is written as 
\bal
1-\mE_{{\rm RNG}, {\rm distill}}(\rho;d)=G^\bullet\qty(\rho, 2^{D_{\rm min}(\phi_d\|\mbF)}),
\eal
where $G^\bullet\qty(\rho, k)$ is defined for $k\geq 1$ as 
\bal
G^\bullet\qty(\rho, k):=\sup\qty{\Tr[\rho W]~|~0\leq W\leq I, W\in\frac{1}{k}\mbF^\bullet}
\eal
and $\mbF^\bullet$ is defined as
\bal
\mbF^\bullet:=\qty{X\in\mL(\mH)~|~\Tr[W\sigma]=1,~\forall\sigma\in\mbF}.
\eal
From this, we have
\bal
1-\mE_{{\rm RNG}, {\rm distill}}(\rho;d)&=\sup\qty{\Tr[\rho W]~|~0\leq W\leq I,~\Tr[\sigma W]= 2^{-D_{\rm min}(\phi_d\|\mbF)},~\forall\sigma\in\mbF}\\
&=\sup\qty{\Tr[\rho W]~|~0\leq W\leq I,~\Tr[\sigma W] =2^{-D_{\rm min}(\phi_d\|\mbF)},~\forall\sigma\in\aff(\mbF)}\\
&=\sup\qty{1-\ve~|~D_{\rm min}(\phi_d\|\mbF)\leq D^\ve_H(\rho\|\aff(\mbF))}.
\eal
From this, we reach the claim.
\end{proof}
The fundamental connection between the one-shot yield of resource distillation and the hypothesis testing divergence is established in Ref.~\cite{Liu_one_shot, Regula_Takagi_2021, lami_2024_asymptotic_quantification}. Even though the proof of Proposition~\ref{pro: one shot distillation} is almost parallel to that in these previous papers, this result extends the previous observation about this connection from a refined viewpoint.

Let us remark on the condition made above.
The former conditions correspond to the case where the set $\mbF$ of free states is full-dimensional, i.e., it holds that $\aff(\mbF)=\mL(\mH)$. In this case, we have $LR_s(\rho)<\infty$. 
Here, it is not clear at first sight whether the given resource theory contains a family $\qty{\phi_d}_{d\in\mbD}$ of reference resource states that satisfy $D_{\rm min}(\phi_d\|\mbF)=LR_s(\phi_d)$.
In Ref.~\cite[Lemma 9]{Takagi_One_shot}, they discuss a general sufficient condition that ensures the existence of such a family of resource states.
Specifically, the condition holds for the Bell state in the bipartite entanglement theory and for the three-qubit Hoggar state in the resource theory of nonstabilizerness~\cite{Takagi_One_shot}.

On the other hand, the latter condition corresponds to the case where the set $\mbF$ of free states is reduced-dimensional, where $\aff(\mbF)\neq\mL(\mH)$ holds. In Ref.~\cite[Theorem 4]{Regula_benchmarking}, it is shown that, in any convex resource theory, a pure state $\phi$ maximizes $D_{\rm max}(\phi\|\mbF)$ if and only if it also maximizes $D_{\rm min}(\phi\|\mbF)$.
The condition $\Tr[\phi\sigma]={\rm const.}$ is an independent condition from the condition above, and is satisfied in, for instance, the resource theory of speakable coherence and the resource theory of thermodynamics.

\subsection{Error exponent and strong converse exponent of general resource distillation}

Now, we will discuss the asymptotic behavior of the optimal error for resource distillation in general resource theories. We would like to investigate how the error in resource distillation behaves when one attempts to distill a resource at a fixed rate $r$. If the target rate $r$ is suboptimal in the sense that $r$ is smaller than the optimal distillable resource, we expect that the fidelity error $\ve_n$ decays to 0 in the asymptotic limit. In this regime, we especially focus on the exponent $c$, called \emph{error exponent}, with which the optimal one-shot error of resource distillation decays $\ve_n\sim 2^{-cn}$. On the other hand, we can also consider the opposite situation in which the target distillation rate $r$ exceeds the optimal distillable resource. In this regime, the error $\ve_n$ does not vanish in the asymptotic limit. Rather, in many cases of resource distillation, the error accumulates to $1$ in the asymptotic limit, which is the direct consequence of the generalized quantum Stein's lemma in Ref.~\cite{Brandao_2010_gqsl, Brand_o_2010, hayashi_generalized_2025, Lami_2025_gqsl}.
Here, we study the smallest exponent $c$, called the \emph{strong converse exponent}, with which the error of the resource distillation accumulates as $\ve_n\sim 1-2^{-cn}$.

The error exponent and the strong converse exponent of the resource distillation are defined below.

\begin{defnboxed}
     \begin{defn}
    Let $(\mbF,\mbO)$ be the resource theory with the family $\qty{\phi_d}_{d\in\mbD}$ of the reference resource states. Suppose that we are given multiple copies of the input state $\rho\in\mD(\mH)$.
    Here, the error exponent $B_{\mbO, {\rm distill},\mfR}(\qty{\rho^{\otimes n}}_{n\in\mbN};r)$ and the strong converse exponent $B^*_{\mbO, {\rm distill},\mfR}(\qty{\rho^{\otimes n}}_{n\in\mbN};r)$ of the resource distillation under the class $\mbO$ of free operations with target distillation rate $r$ with respect to a resource measure $\mfR$ are defined, respectively, as 
    \bal
    B_{\mbO, {\rm distill}, \mfR}(\qty{\rho^{\otimes n}}_{n\in\mbN};r)&:=\sup\qty{\liminf_{n\to\infty}-\frac{1}{n}\log \mE_{\mbO, {\rm distill}}({\rho^{\otimes n}};d_n)~|~\exists \qty{d_n}\in \mbD^\mbN~\mbox{s.t.}~\liminf_{n\to\infty}\frac{1}{n}\mfR(\phi_{d_n})\geq r  },\\
    B^*_{\mbO, {\rm distill},\mfR}(\qty{\rho^{\otimes n}}_{n\in\mbN};r)&:=\inf \qty{\limsup_{n\to\infty}-\frac{1}{n}\log(1-\mE_{\mbO, {\rm distill}}({\rho^{\otimes n}};d_n))~|~\exists \qty{d_n}\in \mbD^\mbN~\mbox{s.t.}~\liminf_{n\to\infty}\frac{1}{n} \mfR(\phi_{d_n})\geq r }.
    \eal
\end{defn}
\end{defnboxed}

Here, note that due to Proposition~\ref{pro: one shot distillation}, we can connect these quantities with the error exponent and the strong converse exponent of the composite hypothesis testing with the i.i.d. sequence $\qty{\rho^{\otimes n}}_{n\in\mbN}$ as the null hypothesis and the sequence $\qty{\mbF_n}_{n\in\mbN}$ as the alternative hypothesis.
Employing the results in Ref.~\cite{fang_2025_error_exponents} exhibited in Eq.~\eqref{eq: composite Hopeffding} combined with Proposition~\ref{pro: one shot distillation}, we have the following.

\begin{proboxed}
\begin{pro}\label{pro: exponent of resource distillation}
    Let $(\mbF,{\rm RNG})$ be the resource theory with a family $\qty{\phi_d}_{d\in\mbD}$ of the reference resource states satisfying $D_{\rm min}(\phi_d\|\mbF)=LR_s(\phi_d)$. Then, the error exponent $B_{{\rm RNG}, {\rm distill}}(\qty{\rho^{\otimes n}}_{n\in\mbN};r)$ and the strong converse exponent $B^*_{{\rm RNG}, {\rm distill}}(\qty{\rho^{\otimes n}}_{n\in\mbN};r)$ of resource distillation under the resource non-generating operations with the resource distillation rate $r$ with respect to the resource measure $\mfR(\cdot)=D_{\rm min}(\cdot\|\mbF)$ is 
    \bal
    B_{{\rm RNG}, {\rm distill},D_{\rm min}}(\qty{\rho^{\otimes n}}_{n\in\mbN};r)&=H_r^{\infty}(\qty{\rho^{\otimes n}}_{n\in\mbN}\|\qty{\mbF_n}_{n\in\mbN})=\lim_{n\to\infty}\frac{1}{n}\inf_{\sigma_n\in\mbF_n} \sup_{\alpha\in (0,1)}\frac{\alpha-1}{\alpha}\qty(nr-\pD_{\alpha}(\rho^{\otimes n}\|\sigma_n)),\\
    B^*_{{\rm RNG}, {\rm distill}, D_{\rm min}}(\qty{\rho^{\otimes n}}_{n\in\mbN};r)&\geq H_r^{*,\infty}(\qty{\rho^{\otimes n}}_{n\in\mbN}\|\qty{\mbF_n}_{n\in\mbN})=\liminf_{n\to\infty}\frac{1}{n}\inf_{\sigma_n\in\mbF_n} \sup_{\alpha>1}\frac{\alpha-1}{\alpha}\qty(nr-\sD_{\alpha}(\rho^{\otimes n}\|\sigma_n)).
    \eal
\end{pro}  
\end{proboxed}

Proposition~\ref{pro: exponent of resource distillation} clarifies the fundamental bound on the precision of resource distillation in the fixed-rate scenario, extending the results in Ref.~\cite{Hayashi_2002_error_exponents, Hayashi_2006_general_formula, lin_2026_exponential_analysis} to the general resource theory.

\section{Strong converse exponent of general resource dilution}

\subsection{One-shot optimal error of general resource dilution}
We consider the fundamental limitation on the precision of the resource dilution in the framework of general resource theories.
First, we consider the one-shot optimal error, where one aims to obtain a target resource state $\rho$ from a fixed reference resource state $\phi_d$.
\begin{defnboxed}
    \begin{defn}
Let $(\mbF,\mbO)$ be the resource theory with the family $\qty{\phi_d}_{d\in\mbD}$ of the reference resource states. The optimal one-shot error of resource dilution from a quantum state $\phi_d$ with the target state $\rho$ is defined as 
\bal
\mE_{\mbO, {\rm dilute}}(\rho;d)&:=1-\max_{\Lambda\in\mbO} F(\Lambda(\phi_d),\rho)
\eal

\end{defn}
\end{defnboxed}

We show that, under an assumption about the family of reference resource states, the one-shot optimal error of general resource dilution can fully be characterized by the logarithmic standard robustness as follows.
\begin{proboxed}
\begin{pro}\label{pro: one-shot dilution}
When the family $\qty{\phi_d}_{d\in\mbD}$ of reference resource states satisfy $D_{\rm min}(\phi_d\|\mbF)=LR_s(\phi_d)$, we have
\bal\label{eq: full dim one shot dilute}
\mE_{{\rm RNG}, {\rm dilute}}(\rho;d)&=\min\qty{\ve~|~D_{\rm min} (\phi_d\|\mbF)\geq LR^{\ve, I, =}_s(\rho) }.
\eal

Furthermore, when the family $\qty{\phi_d}_{d\in\mbD}$ of reference resource states satisfy $D_{\rm min}(\phi_d\|\mbF)=D_{\rm max}(\phi_d\|\mbF)$ and $\Tr[\phi_d\sigma]={\rm const.}$ for any $d\in\mbD$ and any $\sigma\in\mbF$, we have
\bal\label{eq: red dim one shot dilute}
\mE_{{\rm RNG}, {\rm dilute}}(\rho;d)&=\min\qty{\ve~|~D_{\rm min} (\phi_d\|\mbF)\geq D^{\ve, I, =}_{\rm max}(\rho\|\mbF) }.
\eal
\end{pro}
\end{proboxed}

\begin{proof}
We first show Eq.\eqref{eq: full dim one shot dilute}.
Since $\mbF$ is compact, we can take the minimizer $\tilde{\rho},\tau$ for $LR^{\ve,I,=}(\rho)$.
We construct the following CPTP map.
\bal
\Lambda(\omega)=\Tr[\omega \phi_d]\tilde{\rho}+\Tr[\omega(I-\phi_d)]\tau.
\eal
The necessary and sufficient condition for $\Lambda$ to be resource non-generating is formulated as
\bal
\forall \eta\in\mbF,~\frac{1}{1+R^{\ve,I,=}_s(\rho)}&\geq \Tr[\eta\phi_d],\\
\frac{1}{1+R^{\ve,I,=}_s(\rho)}&\geq \max_{\eta\in\mbF}\Tr[\eta\phi_d],\\
2^{-LR^{\ve,I,=}_s(\rho)}&\geq 2^{-D_{\rm min}(\phi_d\|\mbF)},\\
LR^{\ve,I,=}_s(\rho)&\leq D_{\rm min}(\phi_d\|\mbF).
\eal
Also, it is obvious that $\Lambda(\phi_d)=\tilde{\rho}$ holds.
From this, it follows that
\bal
\mE_{{\rm RNG}, {\rm dilute}}(\rho;d)\leq \min\qty{\ve~|~D_{\rm min} (\phi_d\|\mbF)\geq LR^{\ve,I,=}_s(\rho) }.
\eal

Now, we prove the converse part $(\geq)$.
Let $\tilde{\Lambda}\in\mbO$ denote the optimal protocol which achieves the optimization of $\max_{\Lambda\in\mbO}F(\Lambda(\phi_d),\rho)$.
The converse inequality is proven as follows.
\bal
D_{\rm min}(\phi_d\|\mbF)=LR_s(\phi_d)&\geq LR_s(\tilde{\Lambda}(\phi_d))\\
&\geq \min_{\tilde{\rho}\in\mB^{\mE_{{\rm RNG}, {\rm dilute}}(\rho;d)}(\rho),I,=}LR_s(\tilde{\rho})\\
&= LR_s^{\mE_{{\rm RNG}, {\rm dilute}}(\rho;d),I,=}(\rho).
\eal
From this, we can see that 
\bal
\mE_{{\rm RNG}, {\rm dilute}}(\rho;d)\geq \min\qty{\ve~|~D_{\rm min} (\phi_d\|\mbF)\geq LR^{\ve,I,=}_s(\rho) },
\eal
which concludes the proof of Eq.~\eqref{eq: full dim one shot dilute}. 

We next prove Eq.~\eqref{eq: red dim one shot dilute}.
Let us take the minimizer $\ve^*$ of $\min\qty{\ve~|~D_{\rm min} (\phi_d\|\mbF)\geq D_{\rm min}^{\ve^*,I,=}(\rho\|\mbF) }$.
Recall that $D^{\ve^*, I,=}_{\rm max}(\rho\|\mbF)$ is equivalent to the generalized robustness as
\bal
D^{\ve^*,=}_{\rm max}(\rho\|\mbF)=LR^{\ve^*,=}_g(\rho)=\log\min_{\tilde{\rho}\in\mB^{\ve^*,=}(\rho)}\qty{\lambda~|~\frac{\tilde{\rho}+(\lambda-1)\eta}{\lambda}\in\mbF, ~\eta\in\mD(\mH)}.
\eal
From this, there exists $\tilde{\rho}\in\mB^{\ve^*}_{I,=}(\rho)$, $\eta\in\mD(\mH)$ such that 
\bal
\frac{\tilde{\rho}+(2^{D_{\rm min} (\phi_d\|\mbF)}-1)\eta}{2^{D_{\rm min} (\phi_d\|\mbF)}}\in\mbF.
\eal

We define a CPTP map
\bal
\Lambda(\omega):=\Tr[\omega\phi_d]\tilde{\rho}+\Tr[\omega(I-\phi_d)]\eta.
\eal
$\Lambda$ is in fact a free operation, since it holds that for any $\sigma\in\mbF$
\bal
\Lambda(\sigma)&=\Tr[\sigma\phi_d]\tilde{\rho}+\Tr[\sigma(I-\phi_d)]\eta\\
&=2^{-D_{\rm min}(\phi_d\|\mbF)}\tilde{\rho}+(1-2^{-D_{\rm min}(\phi_d\|\mbF)})\eta\in\mbF.
\eal
Moreover, if we feed $\phi_d$, it holds that 
\bal
\Lambda(\phi_d)=\tilde{\rho}\in\mB^{\ve^*}_{I, =}(\rho),
\eal
which implies 
\bal
\mE_{{\rm RNG}, {\rm dilute}}(\rho;d)&\leq \min\qty{\ve~|~D_{\rm min} (\phi_d\|\mbF)\geq D^{\ve, I, =}_{\rm max}(\rho\|\mbF) }.
\eal

We prove the converse part $(\geq )$.
Let $\tilde{\Lambda}\in\mbO$ denote the optimal protocol which achieves the optimization of $\max_{\Lambda\in\mbO}F(\Lambda(\phi_d),\rho)$.
The converse inequality is proven as follows.
\bal
D_{\rm min}(\phi_d\|\mbF)=D_{\rm max}(\phi_d\|\mbF)&\geq D_{\rm max}(\tilde{\Lambda}(\phi_d)\|\mbF)\\
&\geq \min_{\tilde{\rho}\in\mB^{\mE_{{\rm RNG}, {\rm dilute}}(\rho;d)}(\rho),I,=}D_{\rm max}(\rho\|\mbF)\\
&=D^{\mE_{{\rm RNG}, {\rm dilute}}(\rho;d), I,=}_{\rm max}(\rho\|\mbF).
\eal

From this, we have
\bal
\mE_{{\rm RNG}, {\rm dilute}}(\rho;d)&\geq \min\qty{\ve~|~D_{\rm min} (\phi_d\|\mbF)\geq D^{\ve, I, =}_{\rm max}(\rho\|\mbF) },
\eal
which concludes the proof of Eq.~\eqref{eq: red dim one shot dilute}
\end{proof}

From this, we could characterize the error exponent and the strong converse exponent of the resource dilution, but it is, in general, extremely hard to consider the limit of the standard robustness. 
Instead, we consider a relaxed scenario where the available operations are the asymptotic resource non-generating operations, which may create resources from scratch, but the increase in resources can be negligible in the asymptotic limit.
Before going into the details, we first begin by defining the class of available operations in this setting.
\begin{defnboxed}
    \begin{defn}
    Let $\mbF\subset\mD(\mH), \mbF'\subset\mD(\mH')$ be the sets of free states of the input and the output systems, and $\delta>0$ be a positive number. Here, $\delta$-resource-non-generating operation is defined as
    \bal
    {\rm RNG}_{\delta}=\qty{\Lambda:\mD(\mH)\to\mD(\mH')~|~\forall\sigma\in\mbF, ~R_g(\Lambda(\sigma))\leq \delta}.
    \eal
    Furthermore, the sequence of CPTP maps $\qty{\Lambda_n:\mD(\mH_n)\to\mD(\mH'_n)}_{n\in\mbN}$ is called asymptotically resource-non-generating operations if each $\Lambda_n$ is $\delta_n$-resource-non-generating for a sequence $\qty{\delta_n}_{n\in\mbN}$ with $\delta_n\to 0$.
\end{defn}
\end{defnboxed}

For a while, we consider the one-shot resource dilution with $\delta$-resource non-generating operations.
We define the optimal error of the resource dilution $\phi_d\mapsto\rho$ as 
\bal
\mE_{{\rm RNG}_\delta, {\rm dilute}}(\rho;d):=1-\max_{\Lambda\in{\rm RNG}_\delta}F(\Lambda(\phi_d),\rho)
\eal

Note that the class of ${\rm RNG}$ is included in $ {\rm RNG}_\delta$.

The optimal error of the resource dilution under $\delta$-resource non-generating operations is characterized by the generalized robustness, equivalently the smoothed max relative entropy as follows.
\begin{proboxed}
    \begin{pro}\label{pro: converse dilution ARNG}
    Suppose that the family of the reference resource states satisfy $D_{\min}(\phi_d\|\mbF)=D_{\max}(\phi_d\|\mbF)$. Then, the following holds.
    \bal
    \mE_{{\rm RNG}_\delta, {\rm dilute}}(\rho;d)\geq\min\qty{\ve~|~\log(1+\delta)+D(\phi_d\|\mbF)\geq LR_g^{\ve,I,=}(\rho)}
    \eal
\end{pro}
\end{proboxed}

\begin{proof}
    It is shown in \cite[Lemma IV.1]{Brand_o_2010} that for any $\Lambda\in{\rm RNG}_\delta$ and $\rho\in\mD(\mH)$ it holds that
    \bal
    D_{\rm max}(\Lambda(\rho)\|\mbF)\leq  D_{\rm max}(\rho\|\mbF)+\log(1+\delta).
    \eal
    Let $\tilde{\Lambda}$ be the optimal channel which achieves $\ve=\mE_{{\rm RNG}_\delta, {\rm dilute}}(\rho;d)$. Then,
    \bal
    D_{\min}(\phi_d\|\mbF)+\log(1+\delta)=D_{\rm max}(\phi_d\|\mbF)+\log(1+\delta)&\geq D_{\rm max}(\Lambda(\phi_d))\\
    &\geq \min_{\tilde{\rho}\in\mB^{\ve,I,=}(\rho)}LR_g(\tilde{\rho})\\
    &\geq LR_g^{\ve,I,=}(\rho)
    \eal
    holds, which concludes the proof.
\end{proof}

\subsection{Strong converse exponent of the smoothed max relative entropy}
As we will see later, the performance of the resource dilution is tightly connected to quantities such as the standard robustness and the generalized robustness, which is known to be equivalent to the max relative entropy. 
Before discussing the fundamental limitation on the precision of resource dilution, we consider the error exponent and the strong converse exponent of the smoothed max relative entropy.

First, we define the optimal parameter of the smoothing as follows.
\begin{defnboxed}
    \begin{defn}
        Let $d$ be the distance measure, and $(*)$ be either $=$ or $\leq$. The optimal parameter $\ve^{d, (*)}(\rho\|\sigma;\lambda)$ of the smoothed max relative entropy $D^{\ve,d ,(*)}_{\rm max}(\rho\|\sigma)$ with respect to two quantum states $\rho,\sigma\in\mD(\mH)$ is defined as 
        \bal
        \ve^{d, (*)}(\rho\|\sigma;\lambda):=\min\qty{\ve\geq 0~|~D^{\ve, d,(*)}_{\rm max}(\rho\|\sigma)\leq\lambda}.
        \eal
        Similarly, the optimal parameter $\ve^{d, (*)}(\mA\|\mB;\lambda)$ of the smoothed max relative entropy $D^{\ve,d ,(*)}_{\rm max}(\mA\|\mB)$ with respect to two sets $\mA,\mB\subset\mD(\mH)$ of quantum states is defined as
        \bal
        \ve^{d, (*)}(\mA\|\mB;\lambda):=\min\qty{\ve\geq 0~|~D^{\ve, d,(*)}_{\rm max}(\mA\|\mB)\leq\lambda}.
        \eal
    \end{defn}
\end{defnboxed}
From the definition, the optimal parameter $\ve^{d, (*)}(\mA\|\mB;\lambda)$ with respect to two sets $\mA,\mB\subset\mD(\mH)$ of quantum states is reduced to that with respect to the quantum states minimized over $\mA$ and $\mB$, which can be checked as 
\bal
\ve^{d, (*)}(\mA\|\mB;\lambda):&=\min\qty{\ve\geq 0~|~D^{\ve, d,(*)}_{\rm max}(\mA\|\mB)\leq\lambda}\\
&=\min\qty{\ve\geq 0~|~\min_{\substack{\rho\in\mA\\ \sigma\in \mB}}D^{\ve, d,(*)}_{\rm max}(\rho\|\sigma)\leq\lambda}\\
&=\min_{\substack{\rho\in\mA\\ \sigma\in \mB}}\qty{\ve\geq 0~|~D^{\ve, d,(*)}_{\rm max}(\rho\|\sigma)\leq\lambda}\\
&=\min_{\substack{\rho\in\mA\\ \sigma\in \mB}}\ve^{d, (*)}(\rho\|\sigma;\lambda).
\eal

In the following, we discuss the strong converse exponent of the smoothed max relative entropy having two sets of states as the arguments.

\begin{proboxed}
 \begin{pro}\label{pro: lower bound of SC of = smoothed max relative entropy}
    Let $\mA=\qty{\mA_n}_n, \mB=\qty{\mB_n}_n$ be two sequences of sets such that each $\mA_n,\mB_n$ is convex and compact, and both sequences satisfy $\mA_n\otimes \mA_m\subset\mA_{m+n}$ and $\mB_n\otimes \mB_m\subset\mB_{m+n}$.
    Then, the following holds.
    \bal
    \liminf_{n\to\infty}-\frac{1}{n}\log(1-\ve^{I,=}(\mA_n\|\mB_n, nr))\geq \lim_{n\to\infty}\frac{1}{n} \min_{\substack{\rho_n\in\mA_n \\ \sigma_n\in\mB_n}}\sup_{\frac{1}{2}\leq \alpha\leq 1}\frac{\alpha-1}{\alpha}(nr-\sD_\alpha(\rho_n\|\sigma_n)),
    \eal
    which immediately provides a bound for $B_M^{*,I,=}$ defined in \eqref{eq:smoothed max error exponents def} as
    \bal
    B_M^{*,I,=}(\mA\|\mB |r)\geq \lim_{n\to\infty}\frac{1}{n} \min_{\substack{\rho_n\in\mA_n \\ \sigma_n\in\mB_n}}\sup_{\frac{1}{2}\leq \alpha\leq 1}\frac{\alpha-1}{\alpha}(nr-\sD_\alpha(\rho_n\|\sigma_n)).
    \eal
\end{pro}   
\end{proboxed}
Before the proof of this proposition, we remark that the limit
\bal
\lim_{n\to\infty}\frac{1}{n} \min_{\substack{\rho_n\in\mA_n \\ \sigma_n\in\mB_n}}\sup_{\frac{1}{2}\leq \alpha\leq 1}\frac{\alpha-1}{\alpha}(nr-\sD_\alpha(\rho_n\|\sigma_n))
\eal
appearing in the right-hand side indeed exists. To show this, it suffices to show that the sequence
\bal
a_n:=\min_{\substack{\rho_n\in\mA_n \\ \sigma_n\in\mB_n}}\sup_{\frac{1}{2}\leq \alpha\leq 1}\frac{\alpha-1}{\alpha}(nr-\sD_\alpha(\rho_n\|\sigma_n))
\eal
is subadditive, i.e., $a_n+a_m\geq a_{n+m}$ holds for any $n,m\in\mbN$. To this end, we fix arbitrary $\rho_n\in\mA_n, \rho_m\in\mA_m, \sigma_n\in\mB_n,\sigma_m\in\mB_m $.
It holds that
\bal
&\sup_{\frac{1}{2}<\alpha<1}\frac{\alpha-1}{\alpha}\qty((n+m)r-\sD_\alpha(\rho_n\otimes \rho_m\|\sigma_n\otimes \sigma_m))\\
&\quad=\sup_{\frac{1}{2}<\alpha<1}\frac{\alpha-1}{\alpha}\qty(nr-\sD_\alpha(\rho_n\|\sigma_n)+mr-\sD_\alpha(\rho_m\|\sigma_m))\\
&\quad\leq \sup_{\frac{1}{2}<\alpha<1}\frac{\alpha-1}{\alpha}\qty(nr-\sD_\alpha(\rho_n\|\sigma_n))+\sup_{\frac{1}{2}<\alpha<1}\frac{\alpha-1}{\alpha}\qty(mr-\sD_\alpha(\rho_m\|\sigma_m))
\eal
Taking the minimization over $\rho_n\in\mA_n, \rho_m\in\mA_m, \sigma_n\in\mB_n,\sigma_m\in\mB_m $, we have
\bal
a_n+a_m&\geq \min_{\substack{\rho_n\in\mA_n, \rho_m\in\mA_m\\ \sigma_n\in\mB_n,\sigma_m\in\mB_m}}\sup_{\frac{1}{2}<\alpha<1}\frac{\alpha-1}{\alpha}\qty((n+m)r-\sD_\alpha(\rho_n\otimes \rho_m\|\sigma_n\otimes \sigma_m))\\
&\geq \min_{\substack{\rho_{n+m}\in\mA_{n+m}\\ \sigma_{n+m}\in\mB_{n+m}}}\sup_{\frac{1}{2}<\alpha<1}\frac{\alpha-1}{\alpha}\qty((n+m)r-\sD_\alpha(\rho_{n+m}\|\sigma_{n+m}))=a_{n+m}.
\eal
Here, in the last inequality, we used the condition $\mA_n\otimes \mA_m\subset\mA_{m+n}, \mB_n\otimes \mB_m\subset\mB_{m+n}$. 
From Fekete's lemma~\cite{fekete_uber_1923}, we reach the claim.

\begin{proof}[Proof of Proposition~\ref{pro: lower bound of SC of = smoothed max relative entropy} ]

    The key ingredient is Eq.~(L9) in Ref.~\cite{Wang_2019}
    \bal\label{Eq: lower bound of max relative entropy}
    D_{\max}^{\ve ,I,=}(\rho\|\sigma)\geq \sD_\alpha(\rho\|\sigma)+\frac{\alpha}{\alpha-1}\log_2\qty(\frac{1}{1-\ve}), ~\alpha\in[\frac{1}{2},1).
    \eal
Let us fix an arbitrary $\rho_n\in\mA_n$ and $\sigma_n\in\mB_n$.
Recall that $\ve^{I,=}(\rho_n\|\sigma_n,nr)$ is defined as
\bal
\ve^{I,=}(\rho_n\|\sigma_n,nr)=\min\qty{\ve~|~nr\geq D^{\ve, I,=}_{\rm max} (\rho_n\|\sigma_n)}.
\eal
Substituting this $\ve^{I,=}$, for any $\alpha\in[\frac{1}{2},1)$, it holds that
\bal
nr\geq D^{\ve^{I,=}(\rho_n\|\sigma_n, nr), I,=}_{\rm max}(\rho_n\|\sigma_n)&\geq \sD_\alpha (\rho_n\|\sigma_n)+\frac{\alpha}{\alpha-1}\log \qty(\frac{1}{1-\ve^{I,=}(\rho_n\|\sigma_n, nr)}),
\eal
which leads to 
\bal
-\log(1-\ve^{I,=}(\rho_n\|\sigma_n, nr))&\geq \frac{\alpha-1}{\alpha}(rn-\sD(\rho_n\|\sigma_n)).
\eal
Taking supremum over $\alpha$ and taking minimum over $\rho_n\in\mA_n$ and $\sigma_n\in\mB_n$, it holds that 
\bal
\min_{\substack{\rho_n\in\mA_n \\ \sigma_n\in\mB_n }}-\log(1-\ve^{I,=}(\rho_n\|\sigma_n, nr))=-\log(1-\ve^{I,=}(\mA_n\|\mB_n, nr))\geq \min_{\substack{\rho_n\in\mA_n \\ \sigma_n\in\mB_n }}\sup_{\alpha\in[1/2,1)}\frac{\alpha-1}{\alpha}(rn-\sD_\alpha(\rho_n\|\sigma_n))
\eal
Dividing by $n$ and taking the limit ${n\to\infty}$, it holds that 
\bal
\liminf_{n\to\infty}-\frac{1}{n}\log(1-\ve^{I,=}(\mA_n\|\mB_n, nr))\geq \lim_{n\to\infty}\frac{1}{n}\min_{\substack{\rho_n\in\mA_n \\ \sigma_n\in\mB_n }}\sup_{\alpha\in[1/2,1)}\frac{\alpha-1}{\alpha}(rn-\sD_\alpha(\rho_n\|\sigma_n)).
\eal

\end{proof}
Proposition~\ref{pro: lower bound of SC of = smoothed max relative entropy} extends the result in Ref.~\cite{li_operational_2024, Wang_2019} to the composite case.

\subsection{Strong converse exponent of general resource dilution}

We consider the strong converse exponent of resource dilution, the exponent with which the error of the resource dilution converges to $1$. Such a situation occurs when one attempts to obtain the target state from the initial reference state that is smaller than the optimal dimension required for noiseless dilution.
The strong converse exponent of resource dilution is defined as follows.

\begin{defnboxed}
    \begin{defn}
    Let $(\mbF,\mbO)$ be the resource theory with the family $\qty{\phi_d}_{d\in\mbD}$ of the reference resource states. Suppose that we aim to obtain multiple copies of the target state $\rho\in\mD(\mH)$.
    Here, the strong converse exponent of the resource dilution with the rate $r$ with respect to a resource measure $\mfR$ is defined, respectively, as 
    \bal
    B^*_{\mbO, {\rm dilute},\mfR}(\qty{\rho^{\otimes n}}_{n\in\mbN};r)&:=\inf\qty{\limsup_{n\to\infty}-\frac{1}{n}\log(1-\mE_{\mbO, {\rm dilute}}(\rho^{\otimes n};d_n))~|~\exists \qty{d_n}\in \mbD^\mbN~\mbox{s.t.}~\limsup_{n\to\infty}\frac{1}{n}\mfR(\phi_{d_n})\leq r }
    \eal
\end{defn}
\end{defnboxed}

The following proposition gives us a lower bound on the resource dilution under the resource-nongenerating operations.

\begin{proboxed}
    \begin{pro}
    Suppose that the family of the reference resource states satisfies $D_{\min}(\phi_d\|\mbF)=LR_s(\phi_d)$.
    Then, the strong converse exponent of the resource dilution to the target state $\rho$ with the resource-nongenerating operation with rate $r$ with respect to $\mfR(\cdot)=D_{\rm min}(\cdot\|\mbF)$ is bounded as 
    \bal\label{eq: full dim asymptotic dilute}
    B^*_{ {\rm RNG, dilute},D_{\rm min}}(\qty{\rho^{\otimes n}}_{n\in\mbN};r)\geq \limsup_{n\to\infty}\frac{1}{n} \min_{\substack{\rho_n\in\mA_n \\ \sigma_n\in\mB_n}}\sup_{\frac{1}{2}\leq \alpha\leq 1}\frac{\alpha-1}{\alpha}(nr-\sD_\alpha(\rho_n\|\sigma_n)).
    \eal
    Moreover, suppose that the family of the reference resource states satisfies $D_{\min}(\phi_d\|\mbF)=D_{\max}(\phi_d\|\mbF)$ and $Tr[\phi_d\sigma]={\rm const.}$ for any $\sigma\in\mbF$.
    Then, the strong converse exponent of the resource dilution to the target state $\rho$ with the resource-nongenerating operation with rate $r$ with respect to $\mfR(\cdot)=D_{\rm min}(\cdot\|\mbF)$ is characterized as 
    \bal\label{eq: red dim asymptotic dilute}
    B^*_{ {\rm RNG, dilute},D_{\rm min}}(\qty{\rho^{\otimes n}}_{n\in\mbN};r)\geq\limsup_{n\to\infty}\frac{1}{n} \min_{\substack{\rho_n\in\mA_n \\ \sigma_n\in\mB_n}}\sup_{\frac{1}{2}\leq \alpha\leq 1}\frac{\alpha-1}{\alpha}(nr-\sD_\alpha(\rho_n\|\sigma_n)).
    \eal
\end{pro}
\end{proboxed}

\begin{proof}
We first consider the strong converse exponent of resource dilution under the resource non-generating operations.
Since the sequence of the resource non-generating operations is included in the class of asymptotically resource non-generating operations, we have
\bal
B^*_{ {\rm RNG, dilute},D_{\rm min}}(\qty{\rho^{\otimes n}}_{n\in\mbN};r)\geq B^*_{ {\rm ARNG, dilute}, D_{\rm min}}(\qty{\rho^{\otimes n}}_{n\in\mbN};r)
\eal
From this, it suffices to obtain the lower bound of $B^*_{ {\rm ARNG, dilute}, D_{\rm min}}(\qty{\rho^{\otimes n}}_{n\in\mbN};r)$.
To this end, fix an arbitrary sequence $\qty{\delta_n}_n\subset\mbR_{>0}$ of positive numbers which converges to $0$ in the limit $n\to\infty$.
For this sequence, the following holds.
\begin{equation}
    \begin{aligned}
        &\inf_{ \qty{d_n}_n}\qty{\limsup_{n\to\infty}-\frac{1}{n}\log(1-\ve_n)~|~D_{\min}(\phi_{d_n}\|\mbF)\leq rn, \max_{\Lambda\in{\rm RNG}_{\delta_n}}F(\Lambda(\phi_{d_n}),\rho^{\otimes n})\geq 1-\ve_n}\\
=&\inf_{ \qty{d_n}_n}\qty{\limsup_{n\to\infty}-\frac{1}{n}\log(1-\mE_{{\rm RNG}_{\delta_n}, {\rm dilute}}(\rho^{\otimes n};d))~|~ D_{\min}(\phi_{d_n}\|\mbF_{d_n})\leq rn}\\
\geq&\inf_{ \qty{d_n}_n}\qty{\limsup_{n\to\infty}-\frac{1}{n}\log\qty(1-\min\qty{\ve~|~D^{\ve, I,=}_{\max}(\rho^{\otimes n}\|\mbF_n)\leq \log(1+\delta_n)+D_{\min}(\phi_{d_n}\|\mbF_{d_n}),~D_{\min} (\phi_{d_n}\|\mbF_{d_n})\leq rn})}\\
\geq& \limsup_{n\to\infty}-\frac{1}{n}\log\qty(1-\min\qty{\ve~|~D^{\ve, I,=}_{\max}(\rho^{\otimes n}\|\mbF_n)\leq \log(1+\delta_n)+rn})\\
=&\limsup_{n\to\infty}-\frac{1}{n}\log (1-\ve^{I,=}(\rho^{\otimes n}\|\mbF_n, nr+\log(1+\delta_n)))\geq \lim_{n\to\infty}\frac{1}{n} \min_{\substack{\rho_n\in\mA_n \\ \sigma_n\in\mB_n}}\sup_{\frac{1}{2}\leq \alpha\leq 1}\frac{\alpha-1}{\alpha}(nr-\sD_\alpha(\rho_n\|\sigma_n))
    \end{aligned}
\end{equation}
Here, the first inequality follows from Proposition~\ref{pro: converse dilution ARNG}, and the second inequality holds because $D^{\ve, I, =}_{\max }(\rho^{\otimes n}\|\mbF_n)$ is monotonically decreasing in $\ve$.
The last inequality is due to Proposition~\ref{pro: lower bound of SC of = smoothed max relative entropy}.

Now, we prove Eq.~\eqref{eq: red dim asymptotic dilute}. We have
\begin{equation}
    \begin{aligned}
        &\inf_{ \qty{d_n}_n}\qty{\limsup_{n\to\infty}-\frac{1}{n}\log(1-\ve_n)~|~D_{\min}(\phi_{d_n}\|\mbF)\leq rn, \max_{\Lambda\in{\rm RNG}}F(\Lambda(\phi_{d_n}),\rho^{\otimes n})\geq 1-\ve_n}\\
=&\inf_{ \qty{d_n}_n}\qty{\limsup_{n\to\infty}-\frac{1}{n}\log(1-\mE_{{\rm RNG}, {\rm dilute}}(\rho^{\otimes n};d))~|~ D_{\min}(\phi_{d_n}\|\mbF_{d_n})\leq rn}\\
\geq&\inf_{ \qty{d_n}_n}\qty{\limsup_{n\to\infty}-\frac{1}{n}\log\qty(1-\min\qty{\ve~|~D^{\ve, I,=}_{\max}(\rho^{\otimes n}\|\mbF_n)\leq D_{\min}(\phi_{d_n}\|\mbF_{d_n}),~D_{\min} (\phi_{d_n}\|\mbF_{d_n})\leq rn})}\\
\geq& \limsup_{n\to\infty}-\frac{1}{n}\log\qty(1-\min\qty{\ve~|~D^{\ve, I,=}_{\max}(\rho^{\otimes n}\|\mbF_n)\leq rn})\\
=&\limsup_{n\to\infty}-\frac{1}{n}\log (1-\ve^{I,=}(\rho^{\otimes n}\|\mbF_n, nr))\geq \lim_{n\to\infty}\frac{1}{n} \min_{\substack{\rho_n\in\mA_n \\ \sigma_n\in\mB_n}}\sup_{\frac{1}{2}\leq \alpha\leq 1}\frac{\alpha-1}{\alpha}(nr-\sD_\alpha(\rho_n\|\sigma_n)).
    \end{aligned}
\end{equation}
Here, the third line is due to Eq.~\eqref{eq: red dim one shot dilute} in Proposition~\ref{pro: one-shot dilution}, and the last inequality is due to Proposition~\ref{pro: lower bound of SC of = smoothed max relative entropy}.

\end{proof}

\section{Quantum thermodynamics}
So far, we have discussed the fundamental bound on the precision of resource manipulation within the framework of general resource theory. On the other hand, analyses within general resource theory sometimes ignore the specific mathematical structures of individual resource theories. The main reason lies in the assumption that the class of free operations is taken as the resource non-generating operations, which are sometimes much stronger than physically implementable operations.

One such example is the resource theory of thermodynamics, where nonequilibriumness serves as an important resource.
As we discuss below, energy conservation in operations is vital in the thermodynamic setting, a major reason being that it prevents us from creating and detecting energetic coherence, which crucially restricts the power of the available operations. 
On the other hand, the class of operations called the Gibbs-preserving operations is the class of resource non-generating operations in this resource theory, and led to many important results~\cite{Faist_2018_fundamental_work,Wang_2019, Shiraishi_quantum_thermodynamics}. However, it is shown that the Gibbs-preserving operations are strictly stronger than the operationally implementable class called the thermal operations due to the coherence-generating power and the coherence-detecting power of the Gibbs-preserving operations~\cite{Faist2015Gibbs-preserving, Tajima_2024_Gibbs-preserving}.
To characterize the optimal precision of the work extraction, we need to take a close look at the structure of the resource theory of thermodynamics.

\subsection{Preliminaries for quantum thermodynamics}
\subsubsection{Free operations and free state}
We consider a quantum system associated with a Hilbert space $\mH$ and Hamiltonian $H$, in contact with a thermal bath of the inverse temperature $\beta$.

To analyze the amount and precision of work extraction, we employ a resource-theoretic approach, which has proven successful in characterizing the fundamental limitations of various quantum resource-manipulation tasks under certain operational constraints~\cite{Chitamber_Gour, Gour_2025_book}.
In general, a quantum resource theory is defined by identifying the class of quantum channels that can be applied at no cost (called \emph{free operations}) and the sets of states that can be obtained freely (called \emph{free states}). In quantum thermodynamics, the set of free states is a singleton $\qty{\tau}$ containing the Gibbs thermal state $\tau=e^{-\beta H}/\Tr[e^{-\beta H}]$, because we can always obtain the thermal state by leaving the quantum state in contact with the thermal bath and waiting until it gets thermalized.

One standard choice in free operation is the \emph{thermal operation}, which corresponds to an isothermal process for the quantum states. 
A quantum channel $\Lambda:\mD(\mH_A)\to\mD(\mH_B)$ is called a thermal operation if and only if $\Lambda$ is decomposed as (1) appending the thermal state $\tau_E$ of the ancillary system, (2) applying the energy-conserving unitary $U_{AE\to BE'}$, and (3) discarding the subsystem $E'$. Mathematically, a thermal operation is decomposed as
\bal\label{eq: thermal operation}
\Lambda(\rho_A)=\Tr_{E'}[U_{AE\to BE'}(\rho_A\otimes \tau_E)U_{AE\to BE'}^{\dagger}],~~[U_{AE\to BE'},H_A\otimes I_E +I_A\otimes H_E ]=0.
\eal
We define the set of thermal operations as the closure of the set of channels with the Stinespring dilation in Eq.~\eqref{eq: thermal operation}.
Thermal operations satisfy the following important properties. One is that any thermal operation has the thermal state as a fixed point, that is, $\Lambda(\tau_A) = \tau_B$ holds. Another is that the thermal operations are time-translation covariant, i.e., it holds that
\bal
\Lambda(e^{-iH_At}\rho_A e^{iH_A t})=e^{-iH_Bt}\Lambda(\rho_A)e^{-iH_Bt}.
\eal

Since the class of thermal operations is often mathematically involved to deal with, we often consider a larger class of operations called the Gibbs-preserving operations~\cite{Janzing_2000_thermodynamic, Faist_2018_fundamental_work}. Here, the class of Gibbs-preserving operations is the class that contains all the quantum channels satisfying $\Lambda(\tau_A)=\tau_B$. Even though Gibbs-preserving operations have a simple structure and are easy to analyze, they are known to be strictly more powerful than the thermal operations. In fact, it can create and detect energetic coherence from scratch~\cite{Faist2015Gibbs-preserving}, which is prohibited by the time-translation covariance, and it sometimes requires an infinite amount of coherence cost to implement with thermal operations~\cite{Tajima_2024_Gibbs-preserving}.
Here, we also consider the intermediate class of operations called Gibbs-preserving covariant operations, which contains all Gibbs-preserving and time-translation covariant channels.

\subsubsection{Work extraction}
Now, we are in a position to discuss how to quantify the amount and quality of work extraction.
To quantify the work extracted from a quantum state, we consider an additional quantum system, which we call a \emph{work battery}~\cite{Wang_2019, Gour_role_of_quantum_coherence}. Work battery $X_m$ labeled by a positive number $m>0$ is associated with a 2-dimensional Hilbert space $\mH_{X_m}=\Span\qty{\ket{0}_{X_m}, \ket{1}_{X_m}}$ .
We take a Hamiltonian $H_{X}=E_{X,0}\dm{0}+E_{X,1}\dm{1}$ with $E_{X,1}-E_{X,0}=\beta^{-1}\log(m-1)$ so that the thermal state of the battery system $\mu_m$ is $\mu_m=\frac{m-1}{m}\ketbra{0}{0}+\frac{1}{m}\ketbra{1}{1}$.
Suppose that we are given a quantum state $\rho\in\mD(\mH)$. If the conversion $\rho\otimes \mu_m\mapsto\ketbra{1}{1}_\xm$ is possible up to some fidelity error $\ve$ under some free operation $\Lambda\in\mbO$, we say that the work $\beta^{-1}\log m$ is extractable from a quantum state $\rho$ up to the fidelity error $\ve$.
The one-shot optimal extractable work from $\rho\in\mD(\mH)$ under the class of free operation $\mbO$ is defined as 
\bal
\beta W^\ve_{\mbO}(\rho):=\max \qty{m>0~|~\max_{\Lambda\in\mbO}F(\Lambda(\rho\otimes \mu_m),\ketbra{1}{1}_{\xm})\geq 1-\ve}.
\eal
We can also consider the asymptotic (thermodynamic) limit of this quantity by taking multiple copies of the same system. Here, we assume that there is no correlation between subsystems, and the Hamiltonian of the whole system is described as $H^{\times n}:=\sum_{i=1}^nI^{\otimes (i-1)}\otimes H\otimes I^{\otimes (n-i)}$.
The asymptotic limit of the extractable work per copy under the class $\mbO$ of free operations is defined as 
\bal
\beta W^{\rm asymp}_\mbO(\rho):=\lim_{\ve\to +0}\limsup_{n\to \infty}\frac{1}{n}\beta W^\ve_\mbO(\rho^{\otimes n}).
\eal
The previous results studying the work extraction have shown that the asymptotic extractable work from i.i.d. states is characterized as 
\bal
\beta W^{\rm asymp}_\mbO(\rho)=D(\rho\|\tau)
\eal
under any choice of the class $\mbO$ of free operations mentioned above (thermal operations, Gibbs-preserving covariant operations, and Gibbs-preserving operations). 
From this, we can see that whether or not we impose the time-translation covariance on the available operation, this restriction does not influence the extractable work.

\subsubsection{The equivalent definitions of work extraction}
In the following discussion, we adapt the definition of work extraction used in Ref.~\cite{Wang_2019, Gour_role_of_quantum_coherence, Gour_2025_book}.
On the other hand, another definition of work extraction introduced in Ref.~\cite{horodecki_2013} is used in many other papers.
We here explain the definition of work extraction employed in Ref.~\cite{horodecki_2013}, and show that these two definitions are equivalent.
We consider the ancillary system which we call the \emph{work storage} $W$, a qubit system $\mH_W =\qty{\ket{0}_W, \ket{W}_W}$ equipped with the Hamiltonian $H_W=W\ketbra{W}{W}$.
Suppose that we are given a quantum state $\rho$. If the conversion $\rho\otimes \ketbra{0}{0}_W\to \ketbra{W}{W}_W$ is possible with a free operation, we say that we can extract work $W$ from $\rho$.
\begin{lemboxed}
    \begin{lem}[See also Ref.~\cite{watanabe_black_box, Watanabe_universal}]
    Suppose that one can extract work $\beta W=\log  m$ in the sense of the work battery, i.e., $\rho\otimes \mu_\xm\to\ketbra{1}{1}_\xm$, one can also extract work in the sense of the work storage, i.e., $\rho\otimes \ketbra{0}{0}_W\to\ketbra{W}{W}$. Conversely, if one can extract work $\beta W=\log m$ in the sense of work storage, i.e., $\rho\otimes \ketbra{0}{0}\to\ketbra{W}{W}$, we can also extract work in the sense of work battery, i.e., $\rho\otimes \mu_\xm\to\ketbra{1}{1}_\xm$.
    \end{lem}
\end{lemboxed}
\begin{proof}
    We employ the results in Ref.~\cite{horodecki_2013, Wang_2019, Gour_2008_resource_theory_of}: the extractable work---the maximum amount of work extracted from $\rho$--- and the work formation---minimum work cost to create $\rho$---  of an incoherent quantum state $\rho$ with respect to the work storage under thermal operations, Gibbs-preserving covariant operations, and Gibbs-preserving operations are characterized as
    \bal
    \beta W_{\rm ext}(\rho)&=D_{\rm min}(\rho\|\tau),\\
    \beta W_{\rm form}(\rho)&=D_{\rm max}(\rho\|\tau).
    \eal
    We can see that the extractable work and work formation of $\ketbra{1}{1}_\xm$ is $D_{\rm min}(\ketbra{1}{1}_\xm\|\mu_\xm)=D_{\rm max}(\ketbra{1}{1}_\xm\|\mu_\xm)=\log m$.
    
    Suppose that the conversion $\rho\otimes \mu_\xm\to\ketbra{1}{1}_\xm$ is possible. Since the extractable work from $\ketbra{1}{1}_\xm$ with respect to the work storage is $\beta W=\log m$, the following chain of conversion is possible: $\rho\otimes \ketbra{0}{0}_W\to\rho\otimes \mu_\xm\otimes \ketbra{0}{0}_W\to\ketbra{1}{1}_\xm\otimes \ketbra{0}{0}_W\to\ketbra{W}{W}_W$, which proves the first claim.
    
    Now, we show the second claim. Suppose that the conversion $\rho\otimes \ketbra{0}{0}\to\ketbra{W}{W}$. Since the work formation of $\ketbra{1}{1}_\xm$ is $\beta W=\log m$, the following chain of conversion is possible: $\rho\otimes \ketbra{0}{0}_W\to\ketbra{W}{W}_W\to\mu_\xm\otimes \ketbra{W}{W}_W\to\ketbra{1}{1}_\xm\otimes \ketbra{0}{0}_W$. From this, we can at least conclude that the conversion $\rho\to\ketbra{1}{1}_\xm$ is possible in the presence of the exact catalyst $\ketbra{0}{0}_W$.
    
    Here, due to Ref.~\cite[Lemma 1]{Czartowski_2024}, it is shown that the catalyst does not give any advantage when the output state is $\ketbra{1}{1}_\xm$, which implies that the conversion $\rho\to\ketbra{1}{1}_\xm$ is possible without $\ketbra{0}{0}_W$, which concludes the proof.
\end{proof}
We remark that this equivalence holds rigorously in the scenario where one extracts exactly the work. It is not known whether this correspondence also holds for the approximate case, where we allow error on the final state. Thus, the following analysis can depend on the choice of the model that quantifies the extracted work.

\subsection{Reliability functions for work extraction (Proofs of Proposition~\ref{pro: one-shot error of work extraction} and Theorem~\ref{thm: error/strong converse of work extraction})}
Here, we will characterize the optimal precision of work extraction. First, we define the one-shot optimal error of work extraction by focusing on the fidelity error between the final state and the target excited state as follows.
\begin{defnboxed}
    \begin{defn}
        The one-shot optimal error of work extraction from $\rho$ with target amount of work $W$ under the class $\mbO$ of free operations is defined as 
        \bal
\mE_{\mbO,{\rm ext}}(\rho; \beta W):=\min\qty{\ve\geq 0~|~\max_{\Lambda\in\mbO}F(\Lambda(\rho\otimes \mu_{\beta W}),\ketbra{1}{1}_{X_{\beta W}})\geq 1-\ve}.
\eal
    \end{defn}
\end{defnboxed}
We first show that the one-shot optimal error of work extraction is fully characterized by the hypothesis testing divergence.
\begin{proboxed}
    \begin{pro}[Proposition~\ref{pro: one-shot error of work extraction} in the main text]~\label{pro: one-shot error of work extraction app}
    The one-shot optimal error of work extraction from $\rho$ with target amount $W$ of work under Gibbs-preserving operation (GPO), Gibbs-preserving covariant operation (GPC), and thermal operation (TO) is characterized as
    \bal
    \mE_{{\rm GPO, ext}}(\rho; \beta W)&=\min\qty{\ve~|~\beta W\leq D^\ve_H(\rho\|\tau)},\\
    \mE_{{\rm GPC, ext}}(\rho; \beta W)=\mE_{{\rm TO, ext}}(\rho; \beta W)&=\min\qty{\ve~|~\beta W\leq D^\ve_H(\mP_{\tau}(\rho)\|\tau)}.
    \eal
    Here, $\mP_\tau$ is the pinching channel with respect to $\tau$ defined in Eq.~\eqref{eq: spectral pinching}.
\end{pro}
\end{proboxed}
Proposition~\ref{pro: one-shot error of work extraction app} implies that the optimal error of the work extraction is fully characterized by the minimum parameter of the hypothesis testing divergence. An immediate observation is that when the input state $\rho$ and $\tau$ are easier to distinguish, one can extract work from $\rho$ more reliably.

\begin{proof}
    We begin with the case under Gibbs-preserving operations.
    We take the optimal protocol $\Lambda^*\in {\rm GPO}$ for the work extraction $\rho\mapsto \ketbra{1}{1}_{X_m}$. Let $\ve^*:=1-F(\Lambda^*(\rho)\|\ketbra{1}{1}_{X_m})$ be the optimal fidelity error.
    Then, the following holds.
    \bal
    D^{\ve^*}_H(\rho\|\tau)&\geq  D^{\ve^*}_H(\Lambda^*(\rho)\|\tau_{X_m})\\
    &=-\log\min_{\substack{ 0\leq M\leq I \\ \Tr[\Lambda^*(\rho)(I-M)]\leq \ve^* }}\Tr[M\tau_{X_m}].
    \eal
    Here, we used the data-processing inequality of the hypothesis testing divergence. Due to the definition of $\ve^*$, $M=\ketbra{1}{1}$ is feasible, and it holds that $D^{\ve^*}_H(\rho\|\tau)\geq \log m.$ This implies $\mE^D_{m,{\rm GPO}}(\rho)\geq \min\qty{\ve~|~\log m\leq D^\ve_H(\rho\|\tau)}$.

    Now, let us prove the opposite inequality.
    Let us denote the optimal value of the right-hand side as $\ve^*$. Since $D^\ve_H(\rho\|\tau)$ is increasing in $\ve$, $\ve^*$ satisfies $D^{\ve^*}_H(\rho\|\tau)=\log m$. Let $0\leq M\leq I$ denote the optimal measurement for $D^{\ve^*}_H(\rho\|\tau)$.
    Let us take a Gibbs-preserving operation $\Lambda:\mD(\mH)\to\mD(\mH_{X_m})$ defined as
    \bal\label{eq: work extraction protocol in GPO}
    \Lambda(\omega):=\Tr[M\omega]\ketbra{1}{1}_{X_m}+\Tr[(I-M)\omega]\ketbra{0}{0}_{X_m}
    \eal
    This is Gibbs-preserving, which can be checked as follows:
    \bal
    \Lambda(\tau)=\Tr[M\tau]\ketbra{1}{1}_{X_m}+\Tr[(I-M)\tau]\ketbra{0}{0}_{X_m}=\frac{1}{m}\ketbra{1}{1}_{X_m}+\frac{m-1}{m}\ketbra{0}{0}_{X_m}.
    \eal
    From the definition of $\Lambda$, we have $F(\Lambda(\rho),\ketbra{1}{1}_{X_m})\geq 1-\ve^*$, which implies 
    \bal
    \mE^D_{m,{\rm GPO}}(\rho)\leq \min\qty{\ve~|~\log m\leq D^\ve_H(\rho\|\tau)}.
    \eal

    Now we consider the one-shot optimal error under the Gibbs-preserving covariant operations. 
    Here, it holds that~\cite{Gour_role_of_quantum_coherence,watanabe_black_box}
    \bal\label{eq: GPC=GPO+pinching}
    \mE_{{\rm GPC, ext}}(\rho; \beta W)=\mE_{{\rm GPO, ext}}(\mP_{\tau}(\rho); \beta W).
    \eal
    Connecting this with the optimal error under Gibbs-preserving operation, we obtain
    \bal
    \mE_{{\rm GPC, ext}}(\rho; \beta W)=\min\qty{\ve~|~\beta W\leq D^\ve_H(\mP_{\tau}(\rho)\|\tau)}.
    \eal
    
    Finally, we consider the optimal error under thermal operations.
    Since the class of thermal operations is a subclass of the Gibbs-preserving covariant operations, we can easily see that $\mE_{{\rm GPC, ext}}(\rho; \beta W)\leq \mE_{{\rm TO, ext}}(\rho; \beta W)$ holds. Now, we will show the opposite inequality.

    Let us consider the optimal work extraction protocol under Gibbs-preserving covariant operation. From Eq.~\eqref{eq: work extraction protocol in GPO} and Eq.~\eqref{eq: GPC=GPO+pinching}, the optimal work extraction protocol under Gibbs-preserving covariant operation corresponds to the concatenation of the pinching channel $\mP_{\tau}(\cdot)$ and $\Lambda$ in Eq.~\eqref{eq: work extraction protocol in GPO}. If we feed the input state $\rho$, the output state is 
    \bal
     \Lambda(\mP_\tau(\rho)):=\Tr[M\mP_\tau(\rho)]\ketbra{1}{1}_{X_m}+\Tr[(I-M)\mP_\tau(\rho)]\ketbra{0}{0}_{X_m}.
    \eal
    Here, note that thanks to the pinching channel, the pinched state $\mP_\tau(\rho)$ and $\tau$ commute.
    Now, let us take a completely dephasing map $\Delta$ with respect to the simultaneous eigenbasis of $\mP_\tau(\rho)$ and $\tau$. 
    Also, note that this channel is a thermal operation. We now consider a channel 
    $\tilde{\Lambda}:=\Lambda\circ \Delta$. 
    Here, note that $\tilde{\Lambda}$ acts the same on $\rho$ as $\Lambda(\mP_\tau(\cdot))$
    Since $\tilde{\Lambda}$ maps the diagonal state with respect to the simultaneous eigenbasis of $\mP_\tau(\rho)$ and $\tau$ to the diagonal state of $X_m$, we can see this as a classical stochastic matrix. From this, $\tilde{\Lambda}$ can be seen as the Gibbs-preserving operation acting on the classical system where the state is described as the probability distribution, and there is no energetic coherence.
    In Ref.~\cite{horodecki_2013, Shiraishi_2020}, it is shown that Gibbs-preserving operations and the thermal operations are equivalent in the classical system. From this, we can conclude that $\tilde{\Lambda}$ is the thermal operation, and achieves the same performance on work extraction as $\Lambda(\mP_\tau(\cdot))$, which leads $\mE_{{\rm GPC, ext}}(\rho; \beta W)\geq \mE_{{\rm TO, ext}}(\rho; \beta W)$.

\end{proof}

Since Gibbs-preserving covariant operations and the thermal operation are more restrictive than Gibbs-preserving operations, the reliability of work extraction should decrease. Proposition~\ref{pro: one-shot error of work extraction app} clarifies that the decrease in reliability can be represented as the presence of the pinching channel $\mP_\tau$. In fact, from the data-processing inequality of the hypothesis testing divergence, it holds that $\mE_{{\rm GPO, ext}}(\rho; \beta W)\leq \mE_{{\rm GPC, ext}}(\rho; \beta W)=\mE_{{\rm TO, ext}}(\rho; \beta W)$.

We can also consider the asymptotic limit of this quantity. 
The error exponent and the strong converse exponent of the work extraction, where the work extraction rate is kept equal to or larger than $r$ are defined as 
\bal
B_{\mbO,{\rm ext}}(\qty{\rho^{\otimes n}}_{n\in\mbN}; r)&:=\sup_{\{W_n\}_{n\in\mbN}}\qty{\liminf_{n\to\infty}-\frac{1}{n}\log \mE_{\mbO,{\rm ext}}(\rho^{\otimes n}, \beta W_n),~\liminf_{n\to\infty}\frac{1}{n}\beta W_n\geq r},\\
B^*_{\mbO,{\rm ext}}(\qty{\rho^{\otimes n}}_{n\in\mbN}; r)&:=\inf_{\{W_n\}_{n\in\mbN}}\qty{\limsup_{n\to\infty}-\frac{1}{n}\log (1-\mE_{\mbO,{\rm ext}}(\rho^{\otimes n}, \beta W_n)),~\liminf_{n\to\infty}\frac{1}{n}\beta W_n\geq r}.
\label{eq: def of error exponent and strong converse exponent of work extraction}
\eal
Proposition~\ref{pro: one-shot error of work extraction app} enables us to characterize these quantities via the analysis of the error exponent of the hypothesis testing. Especially, if the thermal operations and the Gibbs-preserving covariant operations are taken as the class of free operations, the problem is reduced to determining the error exponent and strong converse exponent when the null hypothesis and the alternative hypothesis are given as the sequence of the pinched states $\qty{\mP_{\tau^{\otimes n}}(\rho^{\otimes n})}_{n\in\mbN}$ and the thermal states $\qty{\tau^{\otimes n}}_{n\in\mbN}$.

The previous results Ref.~\cite{nagaoka_2006_converse_theorem, Hayashi_2007_error_exponent, Mosonyi_Ogawa_2014,Hiai_2008,Hiai_2009, berta_2025_strong_converse} and Ref.~\cite[Lemma 5]{Lipka_Bartosik_2024} about the hypothesis testing with the group symmetry, including the pinched structure, allow us to characterize the error exponent and the strong converse exponent of work extraction as follows.
\begin{thmboxed}
    \begin{thm}[Theorem~\ref{thm: error/strong converse of work extraction} in the main text]\label{thm: error/strong converse of work extraction app}
The error exponent of work extraction from quantum states $\rho\in\mD(\mH)$ with the target work extraction rate $r$ under the classes of the Gibbs-preserving operations, the Gibbs-preserving covariant operations, and the thermal operations is characterized as 
    \bal\label{eq: error exponent of work extraction}
    B_{{\rm GPO},{\rm ext}}(\qty{\rho^{\otimes n}}_{n\in\mbN}; r)&=\sup_{0<\alpha<1}\frac{\alpha-1}{\alpha}\qty(r-\pD_\alpha(\rho\|\tau))\\
    B_{{\rm GPC},{\rm ext}}(\qty{\rho^{\otimes n}}_{n\in\mbN}; r)=B_{{\rm TO},{\rm ext}}(\qty{\rho^{\otimes n}}_{n\in\mbN}; r)&=\sup_{0<\alpha<1}\frac{\alpha-1}{\alpha}\qty(r-\sD_\alpha(\rho\|\tau)).
    \eal
    Moreover, the strong converse exponent of work extraction from  quantum states $\rho\in\mD(\mH)$ with the target work extraction rate $r$ under the classes of the Gibbs-preserving operations, the Gibbs-preserving covariant operations, and the thermal operations is characterized as 
    \bal
    B_{{\rm GPO},{\rm ext}}^*(\qty{\rho^{\otimes n}}_{n\in\mbN}; r)=B_{{\rm GPC},{\rm ext}}^*(\qty{\rho^{\otimes n}}_{n\in\mbN}; r)=B_{{\rm TO},{\rm ext}}^*(\qty{\rho^{\otimes n}}_{n\in\mbN}; r)&=\sup_{\alpha>1}\frac{\alpha-1}{\alpha}\qty(r-\sD_\alpha(\rho\|\tau)).
    \eal
\end{thm}
\end{thmboxed}
\begin{proof}
First, we show the error exponent and the strong converse exponent of work extraction under Gibbs-preserving operations.
Due to Proposition~\ref{pro: one-shot error of work extraction app} and Eq.~\eqref{eq: def of error exponent and strong converse exponent of work extraction}, we can associate the work extraction with those of quantum hypothesis testing as
\bal
B_{{\rm GPO},{\rm ext}} \qty(\qty{\rho^{\otimes n}}_{n\in\mbN}; r)&=\sup\qty{\liminf_{n\to\infty}-\frac{1}{n}\log \mE_{{\rm GPO,  ext}}(\rho^{\otimes n}, \beta W_n),~\liminf_{n\to\infty}\frac{1}{n}\beta W_n\geq r}\\
&=\sup\qty{\liminf_{n\to\infty}-\frac{1}{n}\log \min\qty{\ve_n~|~\beta W_n\leq D^{\ve_n}_H(\rho^{\otimes n}\|\tau^{\otimes n})},~~\liminf_{n\to\infty}\frac{1}{n}\beta W_n\geq r }\\
&=B_H\qty(\qty{\rho^{\otimes n}}_{n\in\mbN}\|\qty{\tau^{\otimes n}}_{n\in\mbN};r),\\
B^*_{{\rm GPO},{\rm ext}}\qty(\qty{\rho^{\otimes n}}_{n\in\mbN}; r)&:=\inf\qty{\limsup_{n\to\infty}-\frac{1}{n}\log (1-\mE_{{\rm GPO,  ext}}(\rho^{\otimes n}, \beta W_n)),~\liminf_{n\to\infty}\frac{1}{n}\beta W_n\geq r}\\
&=\inf\qty{\limsup_{n\to\infty}-\frac{1}{n}\log (1-\min\qty{\ve_n~|~\beta W_n\leq D^{\ve_n}_H(\rho^{\otimes n}\|\tau^{\otimes n})}),~~\liminf_{n\to\infty}\frac{1}{n}\beta W_n\geq r }\\
&=B^*_H\qty(\qty{\rho^{\otimes n}}_{n\in\mbN}\|\qty{\tau^{\otimes n}}_{n\in\mbN};r).
\eal
From the characterization of the exponents of hypothesis testing, recalled in~\eqref{eq:exponents_hyp_testing}, we have
\bal
B_{{\rm GPO},{\rm ext}} \qty(\qty{\rho^{\otimes n}}_{n\in\mbN}; r)&=\sup_{0<\alpha<1}\frac{\alpha-1}{\alpha}\qty(r-\pD_\alpha(\rho\|\tau)),\\
B^*_{{\rm GPO},{\rm ext}} \qty(\qty{\rho^{\otimes n}}_{n\in\mbN}; r)&=\sup_{\alpha>1}\frac{\alpha-1}{\alpha}\qty(r-\sD_\alpha(\rho\|\tau)).
\eal
Here we remark that the exponents of `distinguishability distillation', which is closely related to work extraction under Gibbs-preserving operations, were previously studied in~\cite{Wilde_2022}.

Next, we characterize the error exponent and the strong converse exponent of work extraction under Gibbs-preserving covariant operations and thermal operations.
Following the same argument, they can be associated with the exponents of quantum hypothesis testing as 
\bal
B_{{\rm GPC},{\rm ext}} \qty(\qty{\rho^{\otimes n}}_{n\in\mbN}; r)&=B_{{\rm TO},{\rm ext}} \qty(\qty{\rho^{\otimes n}}_{n\in\mbN}; r)=B_H\qty(\qty{\mP_{\tau^{\otimes n}}(\rho^{\otimes n})}_{n\in\mbN}\|\qty{\sigma^{\otimes n}}_{n\in\mbN};r),\\
B^*_{{\rm GPC},{\rm ext}} \qty(\qty{\rho^{\otimes n}}_{n\in\mbN}; r)&=B^*_{{\rm TO},{\rm ext}} \qty(\qty{\rho^{\otimes n}}_{n\in\mbN}; r)=B^*_H\qty(\qty{\mP_{\tau^{\otimes n}}(\rho^{\otimes n})}_{n\in\mbN}\|\qty{\sigma^{\otimes n}}_{n\in\mbN};r),\\
\eal
from which the problem is reduced to determining the error exponent and the strong converse exponent of the hypothesis testing between $\qty{\mP_{\tau_{\otimes n}}(\rho^{\otimes n})}_{n\in\mbN}$ and $\qty{\tau^{\otimes n}}_{n\in\mbN}$.
Here, we define a function
\bal
\spsi_n(\alpha)&:=\log \Tr\qty((\tau^{\otimes n})^{\frac{1-\alpha}{2\alpha}} \mP_{\tau^{\otimes n}}(\rho^{\otimes n})(\tau^{\otimes n})^{\frac{1-\alpha}{2\alpha}})^\alpha=(\alpha-1)D_\alpha(\mP_{\tau^{\otimes n}}(\rho^{\otimes n})\|\tau^{\otimes n}),\\
\spsi(\alpha)&:=\lim_{n\to\infty}\frac{1}{n}\spsi_n(\alpha)=(\alpha-1)\sD_\alpha(\rho\|\tau).
\eal
Here, the last line is due to the pinching inequality in Eq.~\eqref{eq: pinching inequality}.
In \cite[Lemma III.1]{Mosonyi_2015_two_approaches}, it is shown that $\spsi(\alpha)$ is differentiable in $\alpha\in(0,+\infty)$. From this, thanks to \cite[Theorem 4.8]{Hiai_2008} and \cite[Corollary IV.6]{Mosonyi_2015_two_approaches}, we have
\bal
B_H\qty(\qty{\mP_{\tau^{\otimes n}}(\rho^{\otimes n})}_{n\in\mbN}\|\qty{\sigma^{\otimes n}}_{n\in\mbN};r)&=\sup_{0<\alpha<1}\frac{\alpha-1}{\alpha}(r-\sD_\alpha(\rho\|\tau)),\\
B^*_H\qty(\qty{\mP_{\tau^{\otimes n}}(\rho^{\otimes n})}_{n\in\mbN}\|\qty{\sigma^{\otimes n}}_{n\in\mbN};r)&=\sup_{\alpha>1}\frac{\alpha-1}{\alpha}(r-\sD_\alpha(\rho\|\tau)).\\
\eal

Combining these with the observation above, we reach the conclusion.

\end{proof}

The previous analysis of work extraction shows that the presence or absence of time-translation covariance on the operation does not affect the optimal work extraction rate~\cite{Brandao_resource_theory_of_quantum, Gour_role_of_quantum_coherence}. Accordingly, it is widely recognized that the time-translation covariance of the operations is not critical to asymptotic work extraction. 
Theorem~\ref{thm: error/strong converse of work extraction app} indicates that this observation is not the case.
Specifically, the presence and absence of the time-translation covariance for the error exponent of work extraction appears as the difference of the divergences. In fact, we can see that the error exponent is smaller with the time-translation covariance, which can be easily checked from the fact that $\forall\alpha\in[0,\infty], \sD_\alpha(\rho\|\sigma)\leq \pD_\alpha(\rho\|\sigma)$.

\subsection{Zero-rate error exponent of work extraction (Proof of Theorem~\ref{thm: zero-rate error})}
In the subsequent discussion, we consider the zero-rate error decay of the work extraction, that is, the exponent of the error when we aim at extracting a constant amount of work from an increasing number of quantum states reliably.
The zero-rate error exponent of work extraction is defined as follows.
\bal
B^{\rm Z. R.}_{\mbO,{\rm ext}}(\rho):=\lim_{W\to\infty}\liminf_{n\to\infty}-\frac{1}{n}\log\mE_{\mbO,{\rm ext}}(\rho^{\otimes n}; \beta W)
\eal
The zero-rate exponent of various quantum information processing tasks was recently studied in~\cite{lami_2024_asymptotic_quantification, girardi_2025-1, rippchen_2025_fundamental_quality}.

Employing Proposition~\ref{pro: one-shot error of work extraction app}, the following holds.
\begin{thmboxed}
    \begin{thm}[Theorem~\ref{thm: zero-rate error} in the main text]\label{thm: zero-rate error app}
    The zero-rate error exponent of work extraction under the classes of the Gibbs-preserving operations, the Gibbs-preserving covariant operations, and the thermal operations is characterized as
    \bal
    B^{\rm Z. R.}_{{\rm GPO, ~ ext}}(\rho)&=D(\tau\|\rho),\\
    B^{\rm Z. R.}_{{\rm GPC, ~ ext}}(\rho)=B^{\rm Z. R.}_{{\rm TO, ~ ext}}(\rho)&=D^\star(\tau\|\rho).
    \eal
    Here, $D^\star(\sigma\|\rho)$ is the star divergence defined as 
    \bal
    D^\star(\sigma\|\rho):=\lim_{\alpha\to 1}\lim_{n\to\infty}\frac{1}{n}D_\alpha(\sigma^{\otimes n}\|\mP_{\sigma^{\otimes n}}(\rho^{\otimes n})).
    \eal
\end{thm}
\end{thmboxed}
Theorem~\ref{thm: zero-rate error app} exhibits the stark difference of the zero-rate error under the Gibbs-preserving operation and thermal operation.
In fact, from the data-processing inequality of the Petz Rényi relative entropy, we have
\bal
D(\sigma\|\rho)=\lim_{\alpha\to 1}\pD_\alpha(\sigma\|\rho)\geq \lim_{\alpha\to 1}\frac{1}{n}D_\alpha(\sigma^{\otimes n}\|\mP_{\sigma^{\otimes n}}(\rho^{\otimes n}))=D^\star(\sigma\|\rho).
\eal
Here we remark that the quantity $D^\star(\sigma\|\rho)$ admits an exact expression: it can be understood as the limit $\alpha \to 1$ of the so-called reverse sandwiched R\'enyi relative entropies $\frac{\alpha}{1-\alpha} \sD_{1-\alpha}(\rho\|\sigma)$, for which an analytical, albeit unwieldy, form was established by Audenaert and Datta~\cite{Audenaert_alpha_z_2015}.
The relation of this quantity with pinched hypothesis testing was previously studied in~\cite{Lipka_Bartosik_2024}. In fact, the appearance of the same quantity in quantum state estimation can be traced all the way back to early works by Keyl~\cite{keyl_2006}.

\begin{proof}
    First, we consider the zero-rate error of the work extraction under Gibbs-preserving operations.
    \bal
B^{\rm Z. R.}_{{\rm GPO, ~ ext}}(\rho)&=\lim_{W\to\infty}\liminf_{n\to\infty}-\frac{1}{n}\log\min\qty{\ve_n~|~\Tr[\tau^{\otimes n}M]\leq 2^{-\beta W}~, \Tr[\rho^{\otimes n}(I-M)]\leq\ve_n}\\
&=\lim_{W\to\infty}\liminf_{n\to\infty}\frac{1}{n}D^{2^{-\beta W}}_H(\tau^{\otimes n}\|\rho^{\otimes n})=D(\tau\|\rho).
\eal
The last line follows from quantum Stein's lemma.
Similarly, we can calculate the zero-rate error of the work extraction under Gibbs-preserving covariant operations and thermal operations as
\bal
B^{\rm Z. R.}_{{\rm GPC, ~ ext}}(\rho)=B^{\rm Z. R.}_{{\rm TO, ~ ext}}(\rho)&=\lim_{W\to\infty}\liminf_{n\to\infty}\frac{1}{n}D^{2^{-\beta W}}_H(\tau^{\otimes n}\|\mP_{\tau^{\otimes n}}(\rho^{\otimes n})).
\eal
From this, it suffices to show the following: For any $\rho,\sigma\in\mD(\mH)$, we have
\bal
\lim_{\ve\to +0}\lim_{n\to\infty}\frac{1}{n}D^\ve_H(\sigma^{\otimes n}\|\mP_{\sigma^{\otimes n}}(\rho^{\otimes n}))=D^\star(\sigma\|\rho).
\eal
This is already stated in Ref.~\cite[Lemma 5]{Lipka_Bartosik_2024}, but we present it here for completeness. 
 We employ well-known bounds for the hypothesis testing divergence to obtain the result.
    First, the hypothesis testing divergence $\frac{1}{n}D^\ve_H(\sigma^{\otimes n}\|\mP\qty(\rho^{\otimes n}))$ satisfies  
    \bal
    \frac{1}{n}D^\ve_H(\sigma^{\otimes n}\|\mP\qty(\rho^{\otimes n}))&\geq \frac{1}{n}\pD_\alpha(\sigma^{\otimes n}\|\mP(\rho^{\otimes n}))+\frac{1}{n}\frac{\alpha}{\alpha-1}\log\frac{1}{\ve}\\
    &=\frac{1}{n}D_\alpha(\sigma^{\otimes n}\|\mP(\rho^{\otimes n}))+\frac{1}{n}\frac{\alpha}{\alpha-1}\log\frac{1}{\ve},~~\forall \alpha\in(0,1)
    \eal
    Taking the limit $n\to\infty$ and $\alpha\to 1-0$, we have
    \bal
    \lim_{n\to\infty}\frac{1}{n}D^\ve_H(\sigma^{\otimes n}\|\mP\qty(\rho^{\otimes n}))\geq D^\star(\sigma\|\rho).
    \eal

    Now, we show the opposite inequality. Due to the data-processing inequality of the sandwiched Rényi relative entropy, it holds that, for any $\alpha\in(1,\infty)$~\cite{Cooney_2016_strong_converse},
    \bal
    D^\ve_H(\rho\|\sigma)\leq \sD_\alpha(\rho\|\sigma)+\frac{\alpha}{\alpha-1}\log \frac{1}{1-\ve}
    \eal
    \bal
    \frac{1}{n}D^\ve_H(\sigma^{\otimes n}\|\mP\qty(\rho^{\otimes n}))&\leq \frac{1}{n}\sD_\alpha(\sigma^{\otimes n}\|\mP(\rho^{\otimes n}))+\frac{1}{n}\frac{\alpha}{\alpha-1}\log\frac{1}{1-\ve}\\
    &=\frac{1}{n}\sD_\alpha(\sigma^{\otimes n}\|\mP(\rho^{\otimes n}))+\frac{1}{n}\frac{\alpha}{\alpha-1}\log\frac{1}{1-\ve},~~\forall\alpha\in (1,\infty)
    \eal
Taking the limit $n\to\infty$ and $\alpha\to 1+0$, we have
\bal
\lim_{n\to\infty}\frac{1}{n}D^\ve_H(\sigma^{\otimes n}\|\mP\qty(\rho^{\otimes n}))\leq D^\star(\sigma\|\rho).
\eal
Combining these, we reach the claim.
\end{proof}

\subsection{Separation in the optimal precision of work extraction}
So far, we have seen that the performance of work extraction can depend on the class of allowed operations.
However, the expressions for the error exponents in Theorem~\ref{thm: error/strong converse of work extraction app} and Theorem~\ref{thm: zero-rate error app} involve optimizations, and hence their behavior is not immediately transparent.
In this subsection, we analyze the quantities appearing in Eq.~\eqref{eq: error exponent of work extraction} and then explain how their properties lead to examples in which the optimal error decay rate of work extraction differs among different classes of operations.
We focus on the following two Hoeffding-type quantities, which appear in the expression in Eq.~\eqref{eq: error exponent of work extraction}:
\bal
\pH_r(\rho\|\tau)&:=\sup_{\alpha\in(0,1)}\frac{\alpha-1}{\alpha}\qty(r-\pD_\alpha(\rho\|\tau)),\\
\sH_r(\rho\|\tau)&:=\sup_{\alpha\in(0,1)}\frac{\alpha-1}{\alpha}\qty(r-\sD_\alpha(\rho\|\tau)).
\eal
Then, the claim of Theorem~\ref{thm: error/strong converse of work extraction app} implies that
\bal
B_{{\rm GPO},{\rm ext}}(\qty{\rho^{\otimes n}}_{n\in\mbN}; r)=\pH_r(\rho\|\tau),~~~B_{{\rm TO},{\rm ext}}(\qty{\rho^{\otimes n}}_{n\in\mbN}; r)=\sH_r(\rho\|\tau).
\eal
As mentioned above, from Lieb-Thirring inequality, we have $\pH_r(\rho\|\tau)\geq \sH_r(\rho\|\tau)$.
If we take $r\in(D(\rho\|\tau),\infty)$, the supremum in $\pH_r$ and $\sH_r$ is attained in the limit $\alpha\to 1-0$. Since the Petz Rényi relative entropy and the sandwiched Rényi relative entropy converge to Umegaki relative entropy in the limit $\alpha\to 1$, we have $\pH_r(\rho\|\tau)=\sH_r(\rho\|\tau)=0$ in this regime.
Let us consider the smaller $r$ regime. 
As for $\pH_r(\rho\|\tau)$, when we take $r\in(0,D_{\rm min}(\rho\|\tau))$, it can be seen that $\frac{\alpha-1}{\alpha}(r-\pD_\alpha(\rho\|\tau))$ becomes arbitrarily large in the limit $\alpha\to 0$, which implies that $\pH_r(\rho\|\tau)=+\infty$.
Similarly, if we denote $\sD_0(\rho\|\tau):=\lim_{\alpha\to 0}\sD_\alpha(\rho\|\tau)$ and take $r\in(0, \sD_0(\rho\|\tau))$, we have $\sH_r(\rho\|\tau)=\infty$.

We can show the strict separation between these two quantities as follows.

\begin{proboxed}
    \begin{pro}\label{pro: separation of Hoeffding}
        Let $\rho,\tau$ be two quantum states which satisfies $\supp(\rho)\subset\supp(\tau)$ and $r\in(D_{\rm min}(\rho\|\tau), D(\rho\|\tau))$ be a fixed constant.
        Then, it holds that
        \bal
        \pH_r(\rho\|\tau)&=\sH_r(\rho\|\tau), ~~~(\mbox{if }[\rho,\tau]=0)\\
        \pH_r(\rho\|\tau)&>\sH_r(\rho\|\tau). ~~~(\mbox{if }[\rho,\tau]\neq 0)
        \eal
    \end{pro}
\end{proboxed}
\begin{proof}

First, we discuss the case $[\rho,\tau]=0$.
Due to Ref.~\cite[Corollary 3.4]{Friedland1994}, which discusses the equality condition of the Lieb-Thirring inequality~\cite{LiebThirring}, the Petz and the sandwiched Rényi relative entropy coincide if and only if $\rho$ is incoherent, i.e., $[\rho, \tau]=0$ holds.
From this, for any incoherent state $\rho$, we have
\bal
\pH_r(\rho\|\tau)=\sH_r(\rho\|\tau)
\eal

We consider the case where $[\rho, \tau]\neq 0$ holds.
We first focus on the behavior of $\pH_r(\rho\|\tau)$.
Let us define a function $\pf:(0,1)\to\mbR$ as 
\bal
\pf(\alpha)=\frac{\alpha-1}{\alpha}\qty(r-\pD_\alpha(\rho\|\tau)).
\eal
Note that $\pf$ is continuous on $(0,1)$.
In the limit $\alpha\to +0$, $\pf(\alpha)$ diverges as $\pf(\alpha)\to -\infty$. On the other hand, in the limit $\alpha\to 1-0$, $\pf(\alpha)$ converges to 0.
Since $\pf(\alpha)$ takes a positive value for at least one $\alpha$, the maximum of $\pf(\alpha)$ is achieved in the interior point of the interval $(0,1)$.
Following the same logic, if we take $r\in(\lim_{\alpha\to 0 }\sD_\alpha(\rho\|\tau),D(\rho\|\tau) )$ and define a function $\saf(\alpha)$ as
\bal
\saf(\alpha)=\frac{\alpha-1}{\alpha}\qty(r-\sD_\alpha(\rho\|\tau)),
\eal
the maximum of $\saf(\alpha)$ is achieved in the interior point of the interval $(0,1)$.
Noting that $\lim_{\alpha\to 0 }\sD_\alpha(\rho\|\tau)\leq D_{\rm min}(\rho\|\tau)$ holds, $\saf(\alpha)$ achieves the maximum in the interior point of the interval $(0,1)$ also with $r\in(D_{\rm min}(\rho\|\tau), D(\rho\|\tau))$.

From this, we reach the claim as follows:
Due to Ref.~\cite[Corollary 3.4]{Friedland1994}, it holds that $\sD_\alpha(\rho\|\tau)<\pD_\alpha(\rho\|\tau)$ for any $\alpha\in(0,1)$. If we denote $\alpha^*:=\argmax_{\alpha\in(0,1)}\saf(\alpha)$, then it holds that
\bal
\sH_r(\rho\|\tau)=\max_{\alpha\in(0,1)}\saf(\alpha)=\saf(\alpha^*)<\pf(\alpha^*)\leq \max_{\alpha\in(0,1)}\pf(\alpha)=\pH_r(\rho\|\tau)
\eal
\end{proof}

Proposition~\ref{pro: separation of Hoeffding} implies that the separation between Gibbs-preserving operations and thermal operations is strict if and only if the initial state is coherent.
Combining the observation made above, the reliability function of work extraction in each case is as in FIG.~\ref{fig: reliability function_def}.

\begin{figure}
    \centering
    \includegraphics[width=0.8\linewidth]{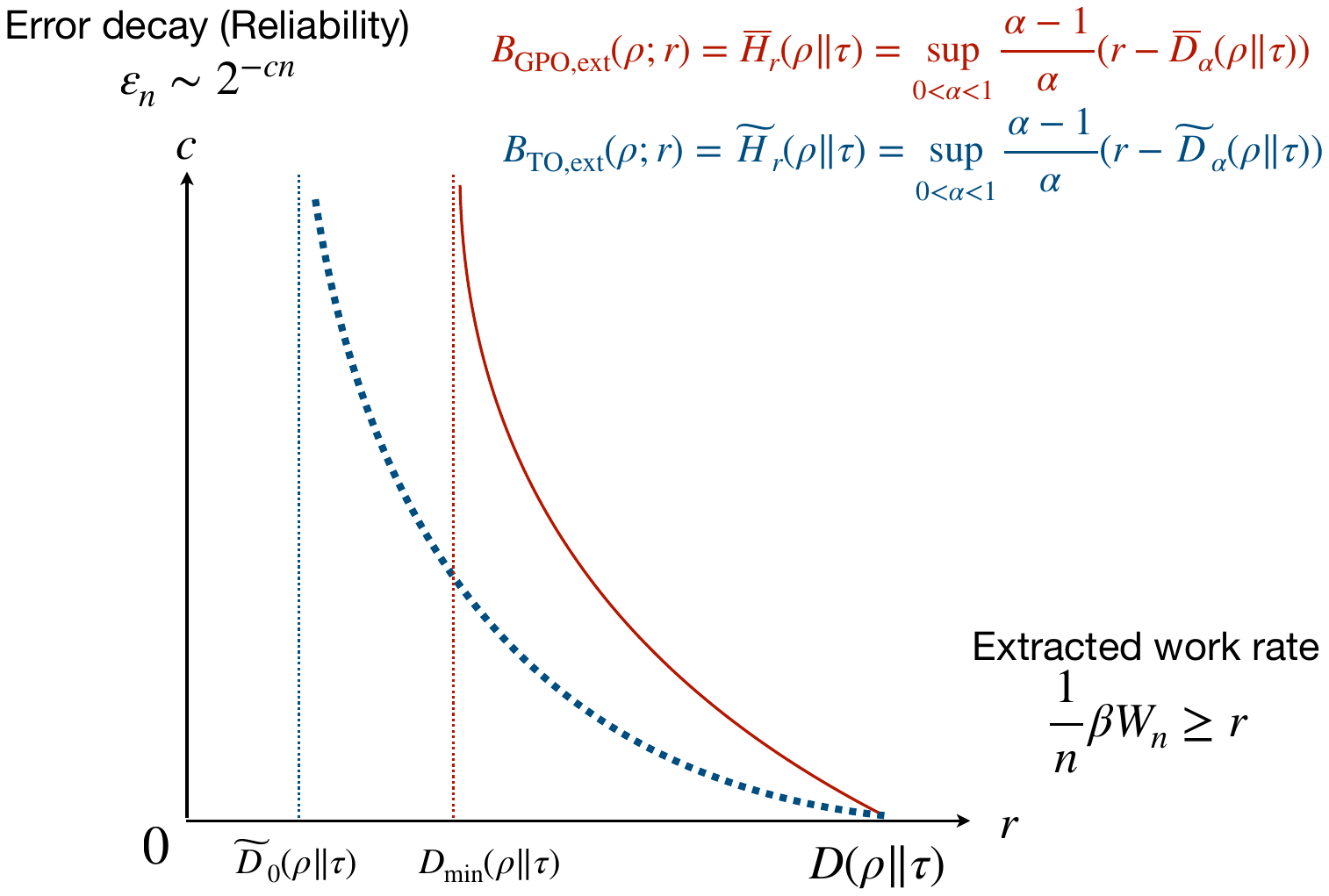}
    \caption{The reliability function of work extraction under the two choices of free operations when the initial state $\rho$ is rank-deficient and coherent. In this case, the reliability function has a strict separation for all rates $D(\rho\|\tau)>r>0$. The reliability function for Gibbs-preserving operations (red line) diverges at the critical rate $r=D_{\rm min}(\rho\|\tau)$, while that for thermal operations (blue dotted line) diverges at the rate $r=\sD_0(\rho\|\tau)$. These two critical rates coincide if and only if $[\Pi_{\supp(\rho)}, \sigma]=0$.}
    \label{fig: reliability function_def}
\end{figure}

Here, we discuss the scenario where the initial state $\rho$ is full-rank. Then, $D_{\rm min}(\rho\|\tau)=0$ holds. Moreover, since $0\leq \sD_0(\rho\|\tau)\leq D_{\rm min}(\rho\|\tau)$ holds , we also have $\sD_0(\rho\|\tau)=0$. In this case, the reliability functions do not diverge. In fact, the reliability function in the limit $r\to 0$ corresponds to the zero-rate error exponent $B^{\rm Z. R.}_{{\rm GPO, ~ ext}}(\rho)=D(\tau\|\rho),~
    B^{\rm Z. R.}_{{\rm GPC, ~ ext}}(\rho)=B^{\rm Z. R.}_{{\rm TO, ~ ext}}(\rho)=D^\star(\tau\|\rho), $ which is finite when $\rho$ is full-rank.
    From these, the reliability function when $\rho$ is coherent and full-rank is given as in FIG.~\ref{fig: reliability function}.
\begin{figure}
    \centering
    \includegraphics[width=0.8\linewidth]{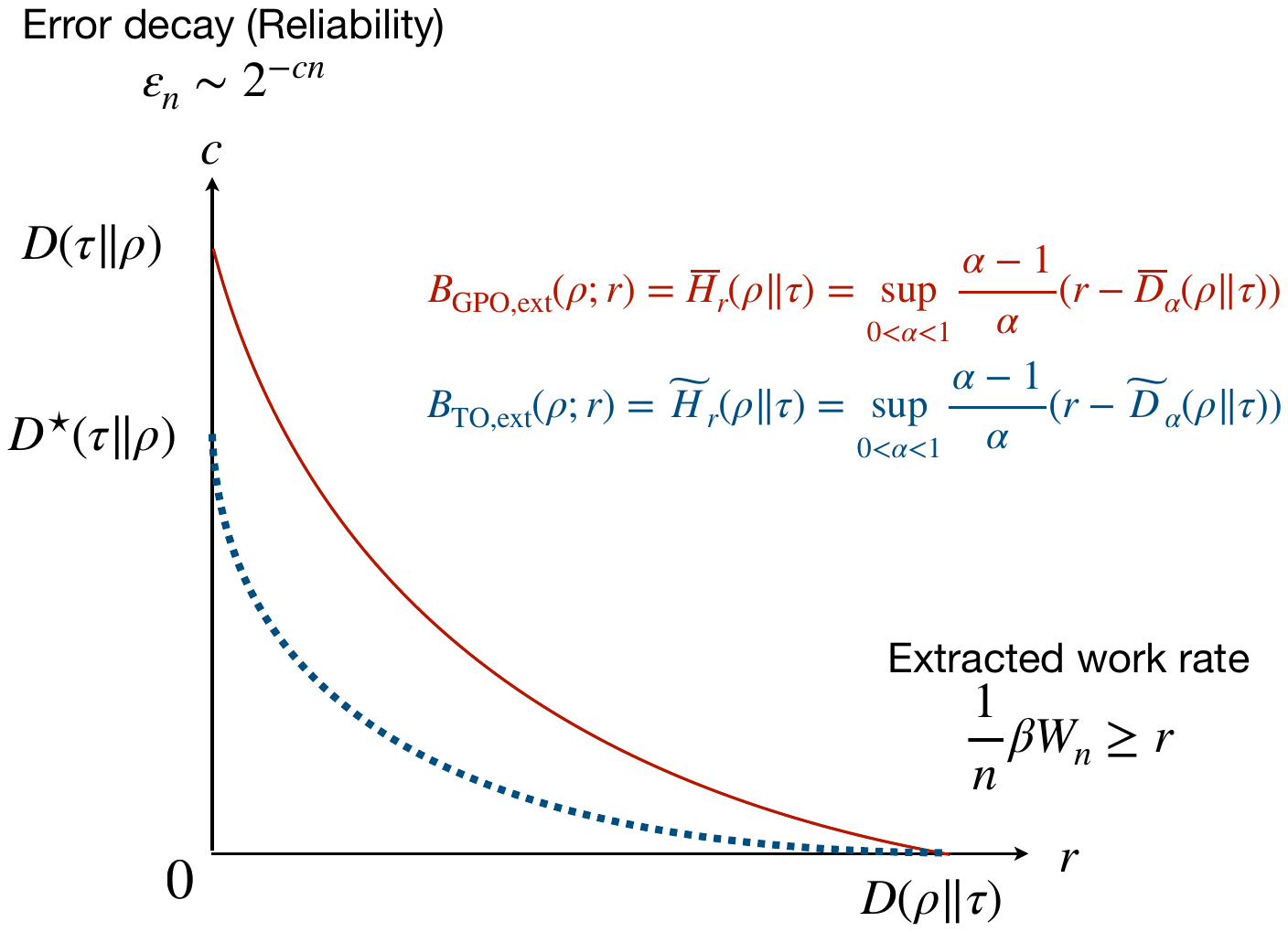}
    \caption{The reliability function of work extraction under the two choices of free operations when the initial state $\rho$ is full-rank and coherent. Also in this case, the reliability function has a strict separation for all rates $D(\rho\|\tau)>r>0$. As we can see from Theorem~\ref{thm: zero-rate error app}, the reliability function under Gibbs-preserving operations (red line) converges to the reversed Umegaki divergence $D(\tau\|\rho)$, while that under thermal operations (blue dotted line) converges to the star divergence $D^\star(\tau\|\rho)$. }
    \label{fig: reliability function}
\end{figure}

For a quantum state $\rho$ satisfying $\sD_0(\rho\|\tau)< D_{\rm min}(\rho\|\tau)$, this separation becomes sharp---the reliability function under Gibbs-preserving operation diverges while that under thermal operations remains finite, when the target extraction rate $r$ is taken as $\sD_0(\rho\|\tau)< r<D_{\rm min}(\rho\|\tau)$.
Let us first provide a simple example such that $\sD_0(\rho\|\tau)< D_{\rm min}(\rho\|\tau)$ is satisfied.
Consider two qubit states $\rho$ and $\tau$, represented in a fixed basis as
\bal
\rho=\mqty(1&0\\0&0),~~\tau=\frac{1}{2}\mqty(1&c\\c&1).
\label{eq: separation example}
\eal
Here, $c$ is a real constant satisfying $0<c<1$.
The state $\tau$ is full-rank and can therefore be regarded, for any inverse temperature $\beta$, as the thermal state of a suitable Hamiltonian.
Ref.~\cite{Datta_a_limit_of_2014} shows that
\bal
\log 2=\pD_0(\rho\|\tau)>\lim_{\alpha\to +0}\sD_\alpha(\rho\|\tau)=\log 2-\log (1+c).
\eal
Therefore, if we choose the target work extraction rate $r$ such that
\bal
\log 2>r>\log 2-\log (1+c),
\eal
then the error exponents of work extraction satisfy
\bal
B_{{\rm GPO},{\rm ext}}(\qty{\rho^{\otimes n}}_{n\in\mbN}; r)&=\infty,\\
    B_{{\rm GPC},{\rm ext}}(\qty{\rho^{\otimes n}}_{n\in\mbN}; r)=B_{{\rm TO},{\rm ext}}(\qty{\rho^{\otimes n}}_{n\in\mbN}; r) &<\infty.
\eal

A natural question to ask is when $\sD_0(\rho\|\tau)$ and $D_{\rm min}(\rho\|\tau)$ coincides. 
In the following, we will derive the necessary and sufficient condition for this.

First, we prove the following two lemmas.

\begin{lemboxed}
    \begin{lem}\label{lem: commuting projector}
        Let $P,Q\in\mL(\mH)$ be two projectors. Then, $PQP$ is also a projector if and only if $[P,Q]=0$ holds.
    \end{lem}
\end{lemboxed}
\begin{proof}
    First, we prove that if $[P, Q]=0$ then $PQP$ is also a projector.
    This can be checked as 
    \bal
    (PQP)^\dagger=PQP,~~PQPPQP=PQPQP=P^2Q^2P=PQP.
    \eal

    We now assume that $PQP$ is a projector.
    Let us take an arbitrary $\ket{v}\in \Im P$. Then, it holds that $PQP\ket{v}=PQ\ket{v}$, which yields 
    \bal
     \|PQP\ket{v}\|^2=\|PQ\ket{v}\|^2.
    \label{eq: norm for vector in image}
    \eal
    Also, since $PQP$ is a projector, i.e., $(PQP)^2=PQP$, we have
    \bal
    \|PQP\ket{v}\|^2=\bra{v}(PQP)^2\ket{v}=\bra{v}PQP\ket{v}=\bra{v}Q\ket{v}=\|Q\ket{v}\|^2.
    \label{eq: norm projected vector}
    \eal
    Writing $Q\ket{v}$ as an orthogonal decomposition $Q\ket{v}=PQ\ket{v}+(I-P)Q\ket{v}$, we have
    \bal
    \|Q\ket{v}\|^2=\|PQ\ket{v}\|^2+\|(I-P)Q\ket{v}\|^2.
    \eal
    Since we have $\|Q\ket{v}\|^2=\|PQ\ket{v}\|^2$ because of \eqref{eq: norm for vector in image} and \eqref{eq: norm projected vector}, it holds that $\|(I-P)Q\ket{v}\|^2=0$, which leads to $(I-P)Q\ket{v}=0$.
    Since $\ket{v}$ is an arbitrary vector in $\Im P$, we have $(I-P)QP=0$.
    Taking the adjoint, we have $PQ(I-P)=0$.
    From these two, we have
    \bal
    QP=PQP,~~PQ=PQP~\Rightarrow~PQ=QP.
    \eal
\end{proof}

The following lemma gives an alternative closed form of $\sD_0(\rho\|\sigma).$
\begin{lemboxed}
    \begin{lem}\label{lem: alternative form of sD0}
        Let $\rho,\sigma\in\mD(\mH)$ be two quantum states.
        Let $\sigma=\sum_{k=1}^m\lambda_k E_k$ be the sprctral decomposition of $\sigma$, where $\lambda_1,\cdots, \lambda_m$ are nonzero distinct eigenvalues satisfying $\lambda_1>\lambda_2>\cdots>\lambda_m>0$, and $F_k$ is defined as $F_k:=E_1+\cdots+E_k$ and $F_0=0$.
        
        Then, it holds that
        \bal\label{eq: explicit form of zero sandwiched from chatgpt}
        \sD_0(\rho\|\sigma)=-\log \sum_{k=1}^m\lambda_k\qty[\rank(\Pi_{\supp(\rho)} F_k \Pi_{\supp(\rho)} )-\rank(\Pi_{\supp(\rho)} F_{k-1} \Pi_{\supp(\rho)} )].
        \eal
    \end{lem}
\end{lemboxed}
\begin{proof}
We consider $\sQ_0(\rho\|\sigma):=2^{-\sD_0(\rho\|\sigma)}$ instead.
We begin with the explicit form of $\sQ_0(\rho\|\sigma)$~\cite{Audenaert_2007, Datta_a_limit_of_2014, Tomamichel_2014_relating_different}
    \bal\label{eq: explicit form of zero sandwiched in Audenaert}
    \sQ_0(\rho\|\sigma)=\max_{i_1,\ldots,i_s}\qty{\sum_{j=1}^s\lambda_{i_j}~|~\qty(\Pi_{\supp(\rho)}\ket{i_j})_{j=1}^s\mbox{ is linearly independent}}
    \eal
    Here, $\sigma=\sum_{i=1}^d\lambda_i\ketbra{i}{i}$ is the spectral decomposition, and $s=\rank(\Pi_{\supp(\rho)}\sigma)$.
    Regarding Eq.~\eqref{eq: explicit form of zero sandwiched in Audenaert}, we need to maximize the sum of $s$ eigenvalues, whose eigenvector has nonzero overlap with the support of $\rho$. To maximize this sum, we first consider the greatest eigenvalue $\lambda_1$. 
    Since the dimension of the common subspace of $\supp(\rho)$ and $\Im(E_1)$ is $\rank(\Pi_{\supp(\rho)} F_1 \Pi_{\supp(\rho)})$, the contribution of $\lambda_1$ to the sum is $\lambda_1\rank(\Pi_{\supp(\rho)} F_1 \Pi_{\supp(\rho)})$. To consider the second greatest eigenvalue $\lambda_2$, one needs to consider the number of eigenvectors corresponding to $\lambda_2$ which has nonzero overlap with $\supp(\rho)$ and which is linearly independent of the vectors in  $\Im(\Pi_{\supp(\rho)}F_1\Pi_{\supp(\rho)})$. 
    Since the dimension of the subspace 
$\Im(\Pi_{\supp(\rho)}F_2\Pi_{\supp(\rho)})\cap\Im(\Pi_{\supp(\rho)}F_1\Pi_{\supp(\rho)})^\perp$
    is $\rank(\Pi_{\supp(\rho)} F_2 \Pi_{\supp(\rho)} )-\rank(\Pi_{\supp(\rho)} F_{1} \Pi_{\supp(\rho)} )$, the contribution of $\lambda_2$ to the sum is $\lambda_2\qty[\rank(\Pi_{\supp(\rho)} F_2 \Pi_{\supp(\rho)} )-\rank(\Pi_{\supp(\rho)} F_{1} \Pi_{\supp(\rho)} )]$.
    Repeating this procedure, we reach the equivalent form in Eq.~\eqref{eq: explicit form of zero sandwiched from chatgpt}.
\end{proof}

\begin{proboxed}
    \begin{pro}\label{pro: separation of sandwich and petz at 0}
        Let $\rho,\sigma$ be quantum states satisfying $\supp(\rho)\subset\supp(\sigma)$. Then, $\sD_0(\rho\|\sigma)=D_{\rm min}(\rho\|\sigma)$ holds if and only if $[\Pi_{\supp(\rho)},\sigma]=0$ holds.
    \end{pro}
\end{proboxed}
We remark that if $[\rho,\sigma]=0,$ then $[\Pi_{\supp(\rho)},\sigma]=0$ also holds, but the opposite direction is not always the case.
\begin{proof}
It suffices to deal with $\pQ_0(\rho\|\sigma):=2^{-D_{\rm min}(\rho\|\sigma)}$ and $\sQ_0(\rho\|\sigma):=2^{-\sD_0(\rho\|\sigma)}$.
 We denote $r_k:=\rank(\Pi_{\supp(\rho)} F_k \Pi_{\supp(\rho)})$, and $t_k:=\Tr[\Pi_{\supp(\rho)} F_k \Pi_{\supp(\rho)}]$. Note that, since $F_k$ is a projector, all the eigenvalues of $\Pi_{\supp(\rho)} F_k \Pi_{\supp(\rho)}$ is in the interval $[0,1]$, which implies that $r_k\geq t_k$.
    The alternative expression of $\sD_0(\rho\|\sigma)$ exhibited in Lemma~\ref{lem: alternative form of sD0} yields
    \bal
    \sQ_0(\rho\|\sigma):= \sum_{k=1}^m \lambda_k(r_k-r_{k-1}),
    \eal
    where we set $r_0=0$ for convention, while $\pQ_0$ is denoted as 
    \bal
    \pQ_0(\rho\|\sigma)&=\Tr[\Pi_{\supp(\rho)}\sigma]\\
    &=\Tr[\Pi_{\supp(\rho)}\sum_{k=1}^m \lambda_kE_k  ]\\
    &=\Tr[\Pi_{\supp(\rho)}\sum_{k=1}^m \lambda_k(F_k-F_{k-1})  ]\\
    &=\sum_{k=1}^m\lambda_k(t_k-t_{k-1}).
    \eal
    From this, it holds that
    \bal
    \sQ_0(\rho\|\sigma)-\pQ_0(\rho\|\sigma)&=\sum_{k=1}^m\lambda_k(r_k-r_{k-1}-t_k+t_{k-1})\\
    &=\sum_{k=1}^{m-1}\qty[(\lambda_k-\lambda_{k+1})(r_k-t_k)]+\lambda_m(r_m-t_m).
    \eal
    Noting that $\lambda_k-\lambda_{k+1}>0$, $\lambda_m>0$, and $r_k\geq t_k$, $\sQ_0(\rho\|\sigma)=\pQ_0(\rho\|\sigma)$ holds if and only if $r_k=t_k$ for all $k$. This equailty holds if and only if $\Pi_{\supp(\rho)} F_k \Pi_{\supp(\rho)}$ is a projector.
    Due to Lemma~\ref{lem: commuting projector}, we have that $\sQ_0(\rho\|\sigma)=\pQ_0(\rho\|\sigma)$ holds if and only if $[\Pi_{\supp(\rho)} ,F_k]=0,~\forall k$. Since $E_k=F_k-F_{k-1} ~(k\geq 2), E_1=F_1$, this condition is equivalent to 
    \bal
    [\Pi_{\supp(\rho)} ,E_k]=0,~\forall k.
    \label{eq: support commutes with eigen projector}
    \eal
    It is straightforward to see that \eqref{eq: support commutes with eigen projector} implies $[\Pi_{\supp(\rho)},\sigma]=0$ recalling that $\sigma=\sum_{k=1}^m\lambda_kE_k$.
    On the other hand, if $[\Pi_{\supp(\rho)},\sigma]=0$ holds, we have 
    \bal
     \sum_j \lambda_j \qty[\Pi_{\supp(\rho)} E_j -  E_j\Pi_{\supp(\rho)}] = 0.
    \eal
    By applying $E_k$ from left and right in both sides, we get \eqref{eq: support commutes with eigen projector}.
    This ensures that \eqref{eq: support commutes with eigen projector} is equivalent to $[\Pi_{\supp(\rho)},\sigma]=0$, showing the claim.
\end{proof}

Proposition~\ref{pro: separation of sandwich and petz at 0} implies the following: Whenever a quantum state does not satisfy $[\Pi_{\supp(\rho)},\tau]=0$, we can always take a work extraction rate $r$ such that the reliability function diverges under the Gibbs-preserving operations but remains finite under the Gibbs-preserving covariant and thermal operations.  
It can be directly checked that the example in \eqref{eq: separation example} indeed violates the commuting condition.

One might expect that the sharp separation where the reliability function under Gibbs-preserving operation diverges while that under thermal operations remains finite can also happen in the zero-rate regime.
However, we can see that this sharp separation is specific to the positive-rate regime, which never occurs in the zero-rate regime. 
For any quantum state $\rho$ that is not full-rank, both the reversed Umegaki divergence and the star divergence $D^\star(\tau\|\rho)$ also diverge, which follows from Ref.~\cite[Theorem 2]{Audenaert_alpha_z_2015}.

\subsection{Data-processing inequality with covariance (Proof of Proposition~\ref{pro: DPI of sandwiched})}
In Eq.~\eqref{eq: error exponent of work extraction}, the optimization of the sandwiched Rényi divergence in the interval $\alpha\in(0,1)$ appears. However, the sandwiched Rényi divergence of order $\alpha\in \qty(0,\frac{1}{2})$ does not satisfy the data-processing inequality~\cite{Berta_2017_on_variational}.
Interestingly, despite this fact, we can show the data-processing inequality of the sandwiched Rényi divergence with a certain covariance.
\begin{proboxed}
    \begin{pro}[Proposition~\ref{pro: DPI of sandwiched} in the main text. Data-processing inequality of sandwiched Rényi divergence under covariant channel]\label{pro: DPI of sandwiched app}
    Let $\sigma\in\mD(\mH)$ be a quantum state and $\mP_\sigma$ be a pinching channel corresponding to $\sigma$ defined in Eq.~\eqref{eq: spectral pinching}.
    Furthermore, let $\Lambda:\mD(\mH_A)\to\mD(\mH_B)$ be a channel satisfying
    \bal
    \mP_{\Lambda(\sigma)^{\otimes n}}\circ \Lambda^{\otimes n}\circ\mP_{\sigma^{\otimes n}}=\mP_{\Lambda(\sigma)^{\otimes n}}\circ \Lambda^{\otimes n}, ~~n\in\mbN.
    \label{eq: DPI condition}
    \eal
    Then, for any quantum state $\rho\in\mD(\mH)$, we have
    \bal
    \sD_\alpha(\rho\|\sigma)&\geq \sD_\alpha(\Lambda(\rho)\|\Lambda(\sigma)), ~~\forall\alpha\in(0,\infty],\\
    D^\star(\sigma\|\rho)&\geq D^\star(\Lambda(\sigma)\|\Lambda(\rho)).
    \eal
\end{pro}
\end{proboxed}
If we take $\sigma$ as the thermal state $\tau$, and $\Lambda$ as a Gibbs-preserving covariant operation or a thermal operation, the condition is satisfied.
\begin{proof}
First, we prove the first inequality.
Since the statement is obvious for $\alpha\in\qty[\frac{1}{2},\infty]$, we focus on $\alpha\in[0,\frac{1}{2}]$.
First, recall Eq.~\eqref{eq: limit of pinched sandwich}
\bal
\sD_\alpha(\rho\|\sigma)=\lim_{n\to\infty}\frac{1}{n}D_\alpha(\mP_{\sigma^{\otimes n}}(\rho^{\otimes n})\|\sigma^{\otimes n}),
\eal
where $D_\alpha$ is the classical Rényi divergence. Here, note that the value of the Petz Rényi divergence and the sandwiched Rényi divergence coincide for commutative states. Due to the data-processing inequality of Petz Rényi divergence, we have
\bal
\sD_\alpha(\rho\|\sigma)&=\lim_{n\to\infty}\frac{1}{n}D_\alpha(\mP_{\sigma^{\otimes n}}(\rho^{\otimes n})\|\sigma^{\otimes n})\\
&=\lim_{n\to\infty}\frac{1}{n}\pD_\alpha(\mP_{\sigma^{\otimes n}}(\rho^{\otimes n})\|\sigma^{\otimes n})\\
&\geq \lim_{n\to\infty}\frac{1}{n}\pD_\alpha(\mP_{\Lambda(\sigma)^{\otimes n}}\circ\Lambda^{\otimes n}\circ\mP_{\sigma^{\otimes n}}(\rho^{\otimes n})\|\Lambda(\sigma)^{\otimes n})\\
&=\lim_{n\to\infty}\frac{1}{n}\pD_\alpha(\mP_{\Lambda(\sigma)^{\otimes n}}\circ\Lambda^{\otimes n}(\rho^{\otimes n})\|\Lambda(\sigma)^{\otimes n})\\
&=\lim_{n\to\infty}\frac{1}{n}\pD_\alpha(\mP_{\Lambda(\sigma)^{\otimes n}}(\Lambda(\rho)^{\otimes n})\|\Lambda(\sigma)^{\otimes n})=\sD_\alpha(\Lambda(\rho)\|\Lambda(\sigma)),
\eal
where in the third line we used the data-processing inequality of Petz Rényi divergence under the channel $\mP_{\Lambda(\sigma)^{\otimes n}}\circ\Lambda^{\otimes n}\circ\mP_{\sigma^{\otimes n}}$ and that $\mP_{\Lambda(\sigma)^{\otimes n}}\circ\Lambda^{\otimes n}\circ\mP_{\sigma^{\otimes n}}(\sigma^{\otimes n})=\Lambda(\sigma)^{\otimes n}$, and in the fourth line we used the assumption \eqref{eq: DPI condition}.
This concludes the proof of the first statement.

We now prove the second inequality.
Since the classical Rényi divergence satisfies $(1-\alpha)D_\alpha(p\|q)=\alpha D_{1-\alpha}(q\|p)$, we have 
\bal
D_\alpha(\sigma^{\otimes n}\|\mP_{\sigma^{\otimes n}}(\rho^{\otimes n}))=\frac{\alpha}{1-\alpha}D_{1-\alpha}(\mP_{\sigma^{\otimes n}}(\rho^{\otimes n})\|\sigma^{\otimes n}).
\eal
Dividing by $n$ and taking the limit $n\to\infty$, we have
\bal
\lim_{n\to\infty}\frac{1}{n}D_\alpha(\sigma^{\otimes n}\|\mP_{\sigma^{\otimes n}}(\rho^{\otimes n}))=\frac{\alpha}{1-\alpha}\sD_{1-\alpha}(\rho\|\sigma).
\eal
From this, it suffices to show that 
\bal
\lim_{\alpha\to 1}\frac{\alpha}{1-\alpha}\sD_{1-\alpha}(\rho\|\sigma)\geq \lim_{\alpha\to 1}\frac{\alpha}{1-\alpha}\sD_{1-\alpha}(\Lambda(\rho)\|\Lambda(\sigma))
\eal
for any $\Lambda$ satisfying the covariance condition.
Indeed, this follows from the first inequality.

\end{proof}

For any Rényi divergence $\mbD_\alpha$ with the data-processing inequality, we can define a resource measure in the resource theory of thermodynamics as 
\bal
\mR_{\mbD_\alpha}(\rho):=\mbD_\alpha(\rho\|\tau).
\eal
Since the sandwiched relative entropy is the minimal divergence of all the divergences satisfying the data-processing inequality and additivity, the resource measure induced from the sandwiched Rényi divergence is no larger than any other resource measure defined in this way. Especially, Proposition~\ref{pro: DPI of sandwiched app} implies that the resource measure $\mR_{\sD_\alpha}$ is monotone under thermal operations and Gibbs-preserving covariant operations for any $\alpha\in(0,\infty]$, while it is not under Gibbs-preserving operations. 
Theorem~\ref{thm: error/strong converse of work extraction app} can be understood as the operational interpretation of the minimal resource measure defined from Rényi divergences.

%TC:endignore
\end{document}